\documentclass[onecolumn,journal,draftclsnofoot,romanappendices]{IEEEtran}

\usepackage{epsfig}
\usepackage{amsmath}
\usepackage{amssymb}
\usepackage{amsfonts}
\usepackage{graphics}
\usepackage{subfigure}

\newtheorem{theorem}{Theorem}
\newtheorem{lemma}{Lemma}

\newtheorem{corollary}{Corollary}
\newtheorem{definition}{Definition}

\DeclareMathOperator*{\argmin}{arg\,min}
\DeclareMathOperator*{\argmax}{arg\,max}

\begin{document}
\title{Relay Selection with Partial Information in Wireless Sensor Networks}
\author{
K. P. Naveen,~\IEEEmembership{Student Member,~IEEE} and
Anurag Kumar,~\IEEEmembership{Fellow,~IEEE}
\thanks{Both the authors are with the {Dept. of Electrical Communication Engineering,}
   	{Indian Institute of Science, Bangalore 560 012, India.}
   	Email:\{naveenkp, anurag\}@ece.iisc.ernet.in}
\thanks{This research was supported in part by a project on Wireless Sensor
Networks for Intrusion Detection, funded by DRDO,
Government of India, and in part by
IFCPAR (Indo-French Center for the Promotion of Advanced Research) (Project 4000-IT-1).}
}
\maketitle

\begin{abstract}
Our work is motivated by geographical forwarding of sporadic alarm
packets to a base station in a wireless sensor network (WSN), where 
the nodes are sleep-wake cycling periodically and asynchronously.
When a node (referred to as the \emph{source}) gets a packet to forward, either by detecting an event 
or from an upstream node, it has to wait for its
neighbors in a \emph{forwarding set} (referred to as \emph{relays}) to wake-up.
Each of the relays is associated with a  
random  reward (e.g., the progress made towards the sink) that is independent and 
identically distributed (iid). To begin with, the source is uncertain about the number of relays, their wake-up times
and the reward values, but knows their distributions. 
At each relay wake-up instant, when a relay reveals its reward value,
the source's problem is to forward the packet
or to wait for further relays to wake-up. In this setting, we seek to minimize the expected
waiting time at the source subject to a lower bound on the average reward. 
In terms of the operations research literature, our work
can be considered as a variant of the \emph{asset selling problem}.
We formulate the relay selection problem as a partially observable Markov decision process
(POMDP), where the unknown state is the number of relays.  
We begin by considering the case where the source knows the number of relays.
For the general case, where the source only knows 
a probability mass function (pmf) on the number of relays,
it has to  maintain a posterior pmf on the number of
relays and forward the packet iff the pmf is in an
\emph{optimum stopping set}.  We show that the optimum stopping set
is convex and obtain an \emph{inner bound} to this set.  We prove a monotonicity
result which yields an \emph{outer bound}. The computational complexity of the 
above policies motivates us to formulate an alternative \emph{simplified} model,
the optimal policy for which is a \emph{simple threshold rule}.
We provide simulation results to compare the performance of the inner and
outer bound policies against the simple policy, and against the optimal
policy when the source knows the exact number of relays.
Observing the simplicity and the good performance of the simple policy, 
we heuristically employ it for \emph{end-to-end packet forwarding} at each hop in a multihop 
WSN of sleep-wake cycling nodes. 
\end{abstract}

\section{Introduction}
\label{introduction}
We are interested in the problem of packet forwarding in a class
of wireless sensor networks (WSNs) in which local inferences based on
sensor measurements could result in the generation of occasional
``alarm'' packets that need to be routed to a base-station, where some
sort of action could be taken \cite{kim-etal09optimal-anycast,cao-etal05optimal-sleep-scheduling, 
premkumar-etal09distributed-detection-localization}. Such a situation could arise, for
example, in a WSN for human intrusion detection or fire detection in a
large region.  Such WSNs often need to run on batteries or on harvested
energy and, hence, must be energy conscious in all their operations.
The nodes of such a WSN would be sleep-wake cycling, waking up
periodically to perform their tasks. One approach for the forwarding
problem is to use a distributed algorithm to schedule the sleep-wake
cycles of the nodes such that the delay of a packet from its source to
the sink on a multihop path is minimized \cite{cao-etal05optimal-sleep-scheduling,lu-etal05delay-efficient}.
An organizational phase is required for such algorithms, which increases the protocol overhead and
moreover the scheduling algorithm has to be rerun periodically since the 
clocks at different nodes drift at different rates (so that the previously computed
schedule would have become stale after long operation time). 
For a survey of routing
techniques in wireless sensor and ad hoc networks and their
classification, see \cite{akkaya-younis05routing-protocols-survey,mauve-etal01survey-position-routing}.

In this paper we are concerned with the sleep-wake cycling approach
that permits the nodes to wake-up independently of each other even though each node
is waking up periodically, i.e.,
asynchronous periodic sleep-wake cycling \cite{naveen-kumar10geographical-forwarding,kim-etal09optimal-anycast}. 
In fact, given the need for a long network life-time, nodes are more likely to be sleeping than awake. 
In such a situation, when a node has a packet to forward, it has to wait
for its neighbors to wake up.  When a neighbor node wakes up, the
forwarding node can evaluate it for its use as a relay, e.g., in terms
of the progress it makes towards the destination node, the quality of
the channel to the relay, the energy level of the relay, etc.,
(see \cite{chang-tassiulas04maximum-lifetime-routing, jinghao-etal06routing-metrics}
for different routing metrics based on the above mentioned quantities).
We think of this as a \emph{reward} offered by the potential relay. 
The end-to-end network objective is to minimize the average total 
delay subject to a lower bound on some measure of total reward along the end-to-end path.
In this paper we address this end-to-end objective by considering optimal
strategies at each hop.
When a node gets a packet to forward, it has to make decisions based only
on the activities in its neighborhood.
Waiting for all potential relays to wake-up and choosing the one with the best
reward maximizes the reward at each hop, but increases the forwarding delay. On the other hand,
forwarding to the first relay to wake-up may result in the loss of the
opportunity of choosing a node with a better reward. Hence, at each hop,
there is a trade-off between the one-hop delay and the one-hop reward. 
By solving the one-hop problem of minimizing the average delay subject to a constraint on 
the average reward, we expect to capture the trade-off between the end-to-end metrics.
For instance, suppose the end-to-end objective is to minimize the expected 
end-to-end delivery delay subject to an upper bound on the expected number of hops in the path, the 
motivation for this constraint being that more hops traversed entails a greater 
expenditure of energy in the network. In our approach, we would heuristically address this problem by considering  at 
each hop the problem of minimizing the mean forwarding delay subject to a lower bound
on the progress made towards the sink. Greater progress at each hop 
entails greater delay per hop, while reducing the number of hops
it takes a packet to reach the sink. 

The local problem setting is the following.
Somewhere in the network a node has just received a packet to forward;
for the \emph{local problem} we refer to this forwarding node as the \emph{source} and think 
of the time at which it gets the packet as $0$. There is an \emph{unknown number of
relays} in the forwarding set of the source. 
In the geographical forwarding context, this lack of information on the number of relays
could model the fact that the neighborhood of a forwarding node could 
vary over time due, for example, to node failures, variation in channel conditions, 
or (in a mobile network) the entry or exit of mobile relays.	
However, we assume that the number of
relays is bounded by a known number $K$, and the source has
an initial probability mass function (pmf), over $(1, \cdots, K)$,
on the number of potential relays.  The source desires to forward the
packet within the interval $[0,T]$, while knowing that the 
relays wake-up independently and uniformly over $[0,T]$ and the
rewards they offer are independently and identically distributed
(iid). 
We will formally introduce our model in Section~\ref{system_model}.
Next we discuss related work and highlight our contributions.
\subsection{Related Work}
\label{related_work}
Here we provide a summary of related literature in the context of 
geographical forwarding and channel selection. Since our problem 
also belongs to the class of asset selling problems
studied in operations research literature,
we survey related work from there as well.\\

\noindent
\emph{\textbf{Geographical forwarding problems:}}
In our prior work \cite{naveen-kumar10geographical-forwarding} we have considered
a simple model where the number of relays is a constant which is known to the source.
There the reward  is simply the progress made by a relay node towards the sink. 
In the current work we have generalized our earlier model by allowing the 
number of relays to be \emph{not known}
to the source. Also, here we allow a general reward structure.

There has been other work in the context of geographical forwarding 
and anycast routing, where the problem of choosing one among several 
neighboring nodes arises. Zorzi and Rao \cite{Zorzi-rao03geographicrandom} consider a scenario
of geographical forwarding in a wireless mesh network in
which the nodes know their locations, and are sleep-wake cycling. They
propose GeRaF (Geographical Random Forwarding), a distributed relaying
algorithm, whose objective is to carry a packet to its destination in
as few hops as possible, by making as large progress as possible at
each relaying stage. For their algorithm, the authors
obtain the average number of hops (for given source-sink distance) as
a function of the node density. These authors do not consider the trade-off between the relay selection
delay and the reward gained by selecting a relay, which is a major contribution
of our work.

Liu et al.\ \cite{liu-etal07CMAC} propose a relay selection approach
as a part of CMAC, a protocol for geographical packet forwarding. With
respect to the fixed sink, a node $i$ has a forwarding set consisting
of all nodes that make progress greater than $r_0$ (an algorithm
parameter). If $Y$ represent the delay until the first wake-up instant
of a node in the forwarding set, and $X$ is the corresponding progress
made, then, under CMAC, node $i$ chooses an $r_0$ that minimizes the
expected normalized latency $\mathbb{E}[\frac{Y}{X}]$. The Random
Asynchronous Wakeup (RAW) protocol \cite{paruchuri-etal04RAW} also
considers transmitting to the first node to wake-up that makes a
progress of greater than a threshold.  Interestingly, this is  the structure
of the optimal policy for our simplified model in \cite{naveen-kumar10geographical-forwarding}.
For the sake of completeness we have described the simplified 
model in this paper as well (see Section~\ref{section:simplified_model}). 
Thus we have provided analytical support
for using such a threshold policy.

Kim et al.\ \cite{kim-etal09optimal-anycast} consider a dense WSN.
Just like the motivation for our model, an
occasional alarm packet needs to be sent, from wherever in the network
it is generated, to the sink. The authors develop an optimal anycast scheme to minimize
average end-to-end delay from any node $i$ to the sink when each
node $i$ wakes up asynchronously with rate $r_i$. They show that
periodic wake-up patterns obtain minimum delay among all sleep-wake
patterns with the same rate. They propose an algorithm
called LOCAL-OPT \cite{kim-etal08tech-report} which yields, for each node $i$, a threshold $h_j^{(i)}$ for 
each of its neighbor $j$. If the time at which neighbor $j$
wakes up is less than $h_j^{(i)}$, then $i$ will transmit to $j$. Otherwise 
$j$ will go back to sleep and $i$ will continue waiting for further neighbors.
A key drawback  is that a  
\emph{configuration phase} is required to run the LOCAL-OPT algorithm.

Rossi et al. \cite{rossi-etal08SARA}, consider the problem where
a node $i$, with a packet to forward and 
which is $n$ hops away from the sink, has to 
choose between two of its shortlisted
neighbors. The first shortlisted neighbor is the one with the least cost 
among all others with hop count $n-1$ (one less than node $i$).
The second one is the least cost node among all its neighbors 
with hop count $n$ (same as that of node $i$). 
Though the first node is on the shortest path, sometimes when its cost
is high, it may not be the best option. 
 It turns out that it is optimal to choose one node over
the other by comparing the cost difference with a threshold. The threshold
depends on the cost distribution of the nodes which are two hops away from node $i$.
Here there is no notion of sleep-wake cycling so that all 
the neighbor costs are known when node $i$
gets a packet to forward. The problem is that of one shot decision making.
In our problem a neighbor's cost will become available only after it wakes up,
at which instant node $i$ has to take decision regarding forwarding.
Hence, ours is a sequential decision problem.\\

\noindent
\emph{\textbf{Channel selection problems:}}
Akin to the relay selection problem is the problem of channel
selection. The authors in \cite{chaporkar-proutiere08joint-probing,
chang-liu07channel-probing} consider a model where there are several
channels available to choose from. The transmitter has to probe the
channels to learn their quality. Probing many channels yields one with
a good gain but reduces the effective time for transmission within the
channel coherence period. The problem is to obtain optimal strategies
to decide when to stop probing and to transmit. 
Here the number of
channels is known and all the channels are 
available at the very beginning of the decision process.  In our
problem the number of relays is not known, and the  relays
become available at random times.\\

\noindent
\emph{\textbf{Asset selling problems:}}
The basic asset selling problem
\cite{Sakaguchi61sequential-sampling,karlin62selling-asset}, comprises
$N$ offers that arrive sequentially over discrete time slots.  The
offers are iid. As the offers arrive, the seller has to decide whether
to take an offer or wait for future offers.  The seller has to pay a
cost to observe the next offer.  Previous offers cannot be recalled.
The decision process ends with the seller choosing an offer.  Over the
years, several variants of the basic problem have been studied, both
with and without recalling the previous offers.  Recently Kang
\cite{Kang05optimal-stopping} has considered a model where a cost has
to be paid to recall the previous best offer.  Further, the previous
best offer can be lost at the next time instant with some probability.
See \cite{Kang05optimal-stopping} for further references to literature on models
with uncertain recall. In \cite{david-levi04optimal-selection}, the
authors consider a model in which the offers arrive at the points of a
renewal process. Additional literature on such work can be found in
\cite{david-levi04optimal-selection}. In these models, either the
number of potential offers is known or is infinite.  In
\cite{ee09asset-selling}, a variant is studied in which the asset
selling process can reach a deadline in the next slot with some fixed
probability, provided that the process has proceeded upto the present
slot.

In our work the number of offers (i.e., relays) is not known. Also the
successive instants at which the offers arrive are the order
statistics of an unknown number of iid uniform random variables over
an interval $[0,T]$. After observing a relay, the probability that
there are no more relays to go (which is the probability that the
present stage is the last one) is not fixed.  This probability has to
be updated depending on the previous such probabilities and the inter
wake-up times between the sucessive relays.  Although our problem
falls in the class of asset selling problems, to the best of our
knowledge the particular setting we have considered in this paper has
not been studied before.

\subsection{Our Contributions}
With the number of relays being unknown, the natural approach is to
formulate the problem as a partially observed Markov decision process
(POMDP). A POMDP is a generalization of an MDP, where at each stage
the actual internal state of the system is not available to the
controller. Instead, the controller can observe a value from an
observation space.  The observation probabilistically depends on the
current actual state and the previous action. In some cases, a POMDP
can be converted to an equivalent MDP by regarding a belief (i.e., a
probability distribution) on the state space as the state of the
equivalent MDP. For a survey of POMDPs see \cite{monahan82survey}.
It is clear that, even if the actual state space is finite, the belief space is
uncountable. There are several algorithms available to obtain the
optimal policy when the actual state space is finite
\cite{lovejoy91algorithmic-methods}, starting from the seminal work by
Smallwood and Sondik \cite{smallwood-sondik73optimal-control}.  When
the number of states is large, these algorithms are computationally
intensive. In general, it is not easy to obtain an optimal policy for a
POMDP. 
In the current work, we have characterized the optimal policy in terms of an 
\emph{optimum stopping set}.
We have made use of the convexity results in
\cite{porta-etal06point-based} and some properties specific to our
problem to obtain an \emph{inner bound} on the optimum stopping set. We prove 
a simple monotonicity result to obtain an \emph{outer bound}. In summary, the following
are the main contributions of our work:
\begin{itemize}
\item We formulate the problem of relay selection with partial
  information as a finite horizon partially observable Markov decision
  process (POMDP), with the unknown state being the actual number of
  relays (Section~\ref{pomdp_formulation}). 
  The posterior pmf on the number of relays is shown to be 
  a sufficient decision statistic.
\item We first consider the completely observable MDP (COMDP) version
  of the problem where the source knows the number of relays with 
  probability one (wp1) (Section~\ref{COMDP_section}). The optimal policy is characterized by a sequence  
  of threshold functions.
\item For the POMDP, at each stage the optimum stopping set is the 
  set of all pmfs on the number of relays where it is optimal to stop (Section~\ref{bounds}).  
  We prove that this set is convex (Section~\ref{convex_section}), 
  and provide an \emph{inner bound} (\emph{subset}) for it (Section~\ref{inner_bound}). 
  We prove a monotonicity result and obtain an \emph{outer bound} (\emph{superset}, Section~\ref{outer_bound}).
  The threshold functions obtained in COMDP version are used in the design of the bounds. 
  These threshold functions need to be obtained recursively
  which is in general, computationally intensive. 
\item The complexity of the above policies motivates us to consider a 
 \emph{simplified} model (Section~\ref{section:simplified_model}). 
  We prove that the optimal policy for this simplified model is a simple threshold rule.
\item Through simulations (Section~\ref{one_hop_performance}) 
  we study the performance comparision of various policies
  with the optimal COMDP policy. The inner bound policy
  performs slighty better than the outer bound policy. The simple policy 
  obtained from the simplified model performs very close to the inner bound.
  Also, we show the poor performance of a naive policy, 
  that assumes the actual number of relays to be simply
  the expected number. 
\item Finally as a heuristic for the end-to-end problem in the 
  geographical forwarding context, we apply the simple policy at 
  each hop and study the end-to-end performance by simulation (Section~\ref{end_to_end_section}).
  We find that it is possible to tradeoff between the expected 
  end-to-end delay and expected number of hops by tuning a parameter.
\end{itemize}
For the ease of presentation, in the main sections we only provide an outline of
the proof for most of the lemmas, followed by a brief description.
Formal proofs are available in Appendices~\ref{COMDP_section_proofs_appendix}, 
\ref{bounds_section_proofs_appendix} and \ref{simplified_model_proofs_appendix}. 
Appendix~\ref{simulation_cases_appendix} contains additional simulation results.



\section{System Model}
\label{system_model}
We consider the one stage problem in which a node in the network receives 
a packet to forward. We call this node the ``source''
 and the nodes that it could potentially  forward the packet to are called
``relays''. The \emph{local problem} is taken to start at time $0$. Thus at time $0$,
the source node has a packet to forward to a sink but needs
a relay node to accomplish this task.  There is a \emph{nonempty} set of $N$
relay nodes, labeled by the indices $1,2,\cdots,N$.  $N$ is a random
variable bounded above by $K$, a system parameter that is known to the source
node, i.e., the support of $N$ is
$\{1,2,\cdots,K\}$. The source does not know $N$, but knows the bound $K$,
and a pmf $p_0$ on $\{1,2,\cdots,K\}$, which is the initial pmf of $N$.  A
relay node $i$, $1\le i\le N$, becomes available to the source at the instant $T_i$.
The source knows that the instants $\{T_i\}$ are iid uniformly
distributed on $(0,T)$. Observe that this would be the case
if the wake-up instants of all the nodes in the network are periodic with period $T$,
if these (periodic) renewal processes are stationary and independent, and if 
the forwarding node's decision instants 
are stopping times w.r.t. these wake-up time processes \cite{wolff89stochastic-modeling}.

We call $T_i$ the \emph{wake-up instant} of
relay $i$.  If the source forwards the packet to the relay $i$, then a
\emph{reward} of $R_i$ is accrued.  The rewards $R_i, i=1,2,\cdots,N$, are
iid random variables with pdf $f_R$.
The support of $f_R$ is $[0,\overline{R}]$.  The source knows this
statistical characterisation of the rewards, and also that the
$\{R_i\}$ are independent of the wake-up instants $\{T_i\}$.  When a
relay wakes up at $T_i$ and reveals its reward $R_i$, the source has to
decide whether to transmit to relay $i$ or to wait for further relays.
\emph{If the source decides to wait, then it instructs the relay with the
best reward to stay awake, while letting the rest go back to sleep}. This way the
source can always forward to a relay with the best reward among those 
that have woken up so far.

Given that $N=n$
(throughout this discussion we will focus on the event $(N=n)$), let $W_1,W_2,\cdots,W_n$ represent the order statistics of
$T_1,T_2,\cdots,T_n$, i.e., the $\{W_k\}$ sequence is the $\{T_i\}$
sequence sorted in the increasing order. 
The pdf of the
$k$ th ($k\le n$) order statistic \cite[Chapter 2]{orderstatistics} is, 
for $0< u< T$, 
\begin{eqnarray}
\label{pdfWk_equation}
 f_{{W_k}|N}(u|n)&=&\frac{n!u^{k-1}(T-u)^{n-k}}{(k-1)!(n-k)!T^n}\mbox{.}
\end{eqnarray}
Also the joint pdf of the  $k$ th and the $\ell$ th 
order statistic (for $k<\ell\le n$) is,
for $0<u\le v<T$,
\begin{eqnarray}
\label{pdfWkWl_equation}
 f_{W_k,W_\ell|N}(u,v|n)=\frac{n!u^{k-1}(v-u)^{\ell-k-1}(T-v)^{n-\ell}}{(k-1)!(\ell-k-1)!(n-\ell)!T^n}\mbox{.}
\end{eqnarray}
Using the above expressions, we can write down the 
conditional pdf $f_{W_{k+\ell}|W_k,N}$ (for $1<\ell\le n-k$) 
as, for $0<w<T$ and $0\le u<T-w$,
\begin{eqnarray}
\label{pdfWlcondWk_equation}
f_{W_{k+\ell}|W_k,N}(w+u|w,n)=\frac{f_{W_k,W_{k+\ell}|N}(w,w+u|n)}{f_{{W_k}|N}(w|n)}\nonumber\\
 =\frac{(n-k)!u^{\ell-1}((T-w)-u)^{(n-k)-\ell}}{(\ell-1)!((n-k)-\ell)!(T-w)^{(n-k)}}\mbox{.}
\end{eqnarray}
Comparing (\ref{pdfWlcondWk_equation}) with (\ref{pdfWk_equation}), as
expected, we observe that, given $N=n$, the pdf of the wake-up instant of the
$(k+\ell)$ th node, conditioned on the wake-up instant of the $k$ th node,
is the $\ell$ th order statistic of $(n-k)$ iid random variables that
are uniform on the remaining time $(T-w)$.  Let $W_0=0$ and define
$U_k=W_k-W_{k-1}$ for $k=1,2,\cdots,n$.  $U_k$ are the \emph{inter-wake-up}
time instants between the consecutive nodes (see Fig.~\ref{wake-up_rewards}).  Later we will be
interested in the conditional pdf $f_{U_{k+1}|W_k,N}$ for
$k=0,1,\cdots,n-1$ which is given by, for $0< w< T$ and $0\le u< T-w$,
\begin{eqnarray}
	\label{condorderstat_equation}
	f_{U_{k+1}|W_k,N}(u|w,n) &=& f_{W_{k+1}|W_k,N}(w+u|w,n)\nonumber\\
	&=& \frac{(n-k)(T-w-u)^{n-k-1}}{(T-w)^{n-k}}\mbox{.}
\end{eqnarray}
 The conditional expectation is given by,
\begin{eqnarray}
 \label{condexpectation_equn}
\mathbb{E}[U_{k+1}|W_k=w,N=n] &=& \frac{T-w}{n-k+1}\mbox{,}
\end{eqnarray}
which is simply the expected value of the minimum of $n-k$ random variables ($n-k$ is the remaining number of relays), 
each of which are iid uniform on the interval $[0,T-w)$ ($T-w$ is the remaining time).
\noindent
\begin{definition}
\label{Defn:expectation}
For notational simplicity we define,
\begin{eqnarray*}
f_{k}(u|w,n) &:=& f_{U_{k+1}|W_k,N}(u|w,n) \\
\mathbb{E}_{k}[\cdot|w,n]&:=&  \mathbb{E}[\cdot|W_k=w,N=n] 
\end{eqnarray*}
Note that $f_{k}(\cdot|w,n)$ depends on $n$ and $k$ 
through the difference $n-k$ and depends on $w$ through $T-w$. 
\hfill $\blacksquare$
\end{definition}

Since the reward sequence $R_1,R_2,\cdots,R_n$
is iid and independent of the wake-up instants $T_1,T_2,\cdots,T_n$,
we write $(W_k,R_k)$ as the pairs of ordered wake-up instants and
the corresponding rewards.
Evidently,
$f_{R_{k+1}|W_k,N}(r|w,n)=f_R(r)$ for $k=0,1,\cdots,n-1$.
Further we define (when $N=n$) $W_{n+1}:=T$, $U_{n+1}:=(T-W_n)$ and $R_{n+1}:=0$. 
Also ${E}_{n}[U_{n+1}|w,n]$ $:=T-w$.
All these variables are depicted in Fig.~\ref{wake-up_rewards}. 
We end this section by listing out, in Table~\ref{glossary_table}, 
most of the symbols that appear
in the paper with a brief description for each.

\begin{figure}[h!]
\centering
\includegraphics[scale=0.36]{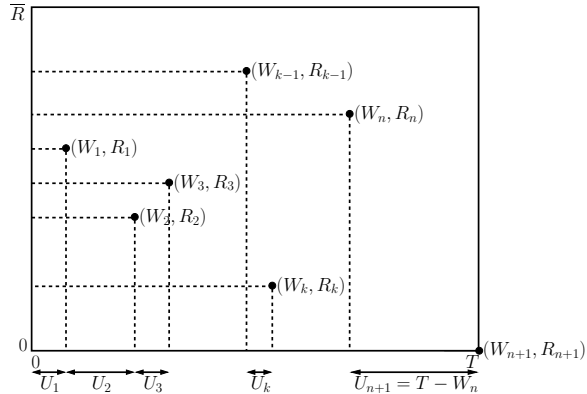}
\caption{There are $N=n$ relays. 
$(W_k,R_k)$ represents the wake-up instant and reward repectively, of the 
$k$th relay. These are shown as points in $[0,T]\times[0,\overline{R}]$. $U_k$ 
are the inter-wake-up times. Note that $W_{n+1}=T$, $R_{n+1}=0$ and $U_{n+1}=T-W_n$.
\label{wake-up_rewards}}
\end{figure}

\begin{table}[h]
\centering 
\begin{tabular}{|p{1.8cm}|p{13.5cm}|}
\hline
\textbf{Symbol} & \textbf{Description}\\
\hline
$\left<a,b\right>$& Inner product of vectors $a$ and $b$\\
\hline
$a_{k}^{\ell}(w,b)$ \hspace{2mm} $b_{k}^{\ell}(w,b)$   & Thresholds lying on the line joining $p_k^{(k)}$ and $p_k^{(k+\ell)}$ 
of the simplex $\mathcal{P}_k$; Used in the construction of the  inner and outer bounds, respectively  \\
\hline
$B_k$ & Best reward so far, i.e., $B_k=\max\{R1,\cdots,R_k\}$ \\
\hline
$c_k(p,w,b)$ & Average cost of continuing at stage $k$ when the  state is $(p,w,b)$\\
\hline
$\mathcal{C}_{k}(w,b)$ & Optimum stopping set at stage $k$ when $(W_k,B_k)=(w,b)$\\
\hline
$\underline{\mathcal{C}}_k(w,b)$ & Inner bound for the stopping set $\mathcal{C}_k(w,b)$\\
\hline
$\overline{\mathcal{C}}_k(w,b)$ & Outer bound for the stopping set $\mathcal{C}_k(w,b)$\\
\hline
$\mathcal{C}_{1step}$& One-step-stopping set for the simplified model\\
\hline
$\mathbb{E}_k[\cdot|w,n]$ & Expectation conditioned on $(W_k,N)=(w,n)$\\ 
\hline
$f_k(\cdot|w,n)$ & pdf of $U_{k+1}$ conditioned on $(W_k,N)=(w,n)$\\
\hline
$f_R(\cdot)$ & pdf of the iid rewards $\{R_k\}$ \\
\hline
$J_k(p,w,b)$ & Optimal cost-to-go function at stage $k$ when the  state is $(p,w,b)$\\
\hline
$K$ & Bound on the number of relays \\
\hline
$N$ & Number of relays; random variable taking values  from $\{1,2,\cdots,K\}$ \\
\hline
$\tilde{N}$ & Number of relays in the simplified model; a constant \\
\hline
$\mathbb{P}(A)$& Probability of an event $A$\\
\hline
$\mathcal{P}_k$ & Set of all pmfs on the set $\{k,k+1,\cdots,K\}$\\
\hline
$(p,w,b)$ & Represents a typical state at stage $k$ where $p\in\mathcal{P}_k$  is the belief state and
		$(W_k,B_k)=(w,b)$\\
\hline
$p_k^{(n)}$ & A corner point in $\mathcal{P}_k$, i.e., $p_k^{(n)}(n)=1$\\
\hline
$R_k$ & Reward of the $k$ th relay \\
\hline
$U_{k+1}$ & Inter wake-up time between the $k+1$ and $k$ th  relay, i.e., $U_{k+1}=W_{k+1}-W_{k}$\\
\hline
$W_k$ & Wake-up instant of the $k$ th relay \\
\hline
$\tilde{W}_k,\tilde{R}_k,\tilde{U}_{k+1}$ & Quantities, analogous to the ones in  the exact model, for the simplified model \\
\hline
$\alpha$ & Threshold obtained from the simplified model\\
\hline
$\gamma$ & Reward constraint for the  problem  in (\ref{main_problem})  \\
\hline
$\delta_{n-k}(w,b)$ & When $p\in\mathcal{P}_k$ is such that $p(k)+p(n)=1$ then it is 
optimal to stop iff $p(n)\le\delta_{n-k}(w,b)$ \\
\hline
$\eta$ &  Lagrange multiplier, see (\ref{unconstrained_problem})  \\
\hline
$-\eta b$ & Average cost of stopping at stage $k$ when $B_k=b$\\
\hline
$\tau_{k+1}(p,w,u)$ & Belief transition function; $\tau_{k+1}(p,w,u)$ is a pmf in  $\mathcal{P}_{k+1}$
	 for a given $p\in\mathcal{P}_k$, $W_k=w$ and $U_{k+1}=u$\\
\hline
$\phi_{n-k}(w,b)$ & Threshold obtained from the COMDP version of the problem; If the source knows wp1 that $N=n$,  then
at some stage $k\le n$ with $(W_k,B_k)=(w,b)$  it is optimal to stop iff $b\ge\phi_{n-k}(w,b)$\\
\hline
\end{tabular}
\caption{List of mathematical notation.}
\label{glossary_table}
\vspace*{-10mm}
\end{table}

\section{The Sequential Decision Problem}
\label{pomdp_formulation}
For the model set up in Section~\ref{system_model}, we now consider the 
following sequential decision problem. At each instant that a relay
wakes up, i.e., $W_1,W_2,\cdots$, the source has to make the decision to
forward the packet, or to hold the packet until the next wake-up instant.
Since the number of available relays, $N$, is unknown, we have a decision 
problem with partial information. We will show how the problem can be set
up in the framework of a partially observable Markov decision process (POMDP)
\cite{monahan82survey} \cite[Chapter 5]{optimalcontrol}.
\subsection{Actions, State Space, and State Transition}
\noindent
\emph{Actions:}
We assume that 
the time instants at which the relays wake-up, i.e., $W_1,W_2,\cdots$, constitute the 
\emph{decision instants or stages}
\footnote{A better choice for the decision instants may be to allow
the source to take decision at any time $t\in(0,T]$. When $N$ is known to the source
it can be argued that it is optimal to take decisions only at relay wake-up instances.
However this may not hold for our case where $N$ is unknown. In this paper we proceed 
with our restriction on the decision instants and consider the general case as a topic for future work.}.
At each decision instant, there 
are two actions possible at the source, denoted $0$ and $1$, where
\begin{itemize}
\item $0$ represents the action to \emph{continue} waiting for more relays to wake-up, and
\item $1$ represents the action to \emph{stop} and forward the packet 
to the relay that provides the best reward
among those that have woken up to the current decision epoch.
\end{itemize}
Since there can be at most $K$ relays, the total number of decision instants is $K$.
The decision process technically ends at the first instant $W_k$, at which the source
chooses action $1$, in which case we assume that all the subsequent decision instants, 
$k+1,\cdots,K$, occur at $W_k$.
In cases where the source ends up waiting until time $T$ 
(referring to Fig.~\ref{wake-up_rewards}, this is possible if, even at $W_n$ the 
source decides to continue, not realizing that it has seen all the relays there 
are in its forwarding set), all the subsequent decision
instants are assumed to occur at $T$.\\
 
\noindent
\emph{State Space:}
At stage $0$ the
state space is simply $\mathcal{S}_0^{a}=\Big\{(n,0,0): 1\le n\le K\Big\}$ 
and the only action possible is $0$, where $a$ in the superscript is to signify that $\mathcal{S}_0^{a}$ is the
set of \emph{actual} internal states of the system. 
The state space at stage $1$ is,
\begin{equation*}
 \mathcal{S}_1^{a}=\Big\{(n,w,b): 1\le n\le K, w\in(0,T), b\in[0,\overline{R}]\Big\}
\end{equation*}
and for stages $k=2,3,\cdots,K$ is,
\begin{eqnarray}
\label{actual_state_space_equn}
 \mathcal{S}_k^{a}&=&\Big\{(n,w,b): k\le n\le K, w\in(0,T), b\in[0,\overline{R}]\Big\}\nonumber\\
&& 	\cup \Big\{(k-1,T,b):b\in[0,\overline{R}]\Big\}\cup \{\psi\}\\
&=&\mathcal{S}_k^{a}(1)\cup\mathcal{S}_k^{a}(2)\cup\mathcal{S}_k^{a}(3)\mbox{.}\nonumber
\end{eqnarray}
 Thus the state space at 
stage $k=2,3,\cdots,K$ is written as the union of three sets. The physical
meanings of  these sets are as follows:
\begin{itemize}
 \item $\mathcal{S}_k^{a}(1)$:
$n$ in the state
triple $(n,w,b)$ represents the actual number of relays. The states in this set 
correspond to the case where there are more than or equal to $k$ relays, i.e., $n$ satisfies,
$k\le n\le K$. In the pair $(w,b)$, $w$ is the wake-up instant ($W_k$)
of the $k$ th relay, and $b$ is the best reward
$(B_k=\max\{R_1,\cdots,R_k\})$ among the relays seen so far. Same remark holds
for the states in $\mathcal{S}_1^{a}$. Stage $0$ begins at
time $0$ with $0$ reward. Hence the states in $\mathcal{S}_0^{a}$ are of the form $(n,0,0)$.

\item $\mathcal{S}_k^{a}(2)$: Suppose there were $k-1$ relays and, at stage $k-1$ the source
decides to continue. Note that it is possible for the source to take such a decision, since 
it \emph{does not know} the number of relays. In such a case, the source ends up waiting until time $T$
and enters stage $k$.
Hence the states in this set are of the form $(k-1,T,b)$ where $b$ represents the best reward among all
the $k-1$ relays ($B_{k-1}$).

\item $\mathcal{S}_k^{a}(3)$: $\psi$ is the \emph{terminating} state. The state at stage $k$ will be
$\psi$, if the source has already forwarded the packet at an earlier
stage.\\
\end{itemize}

\noindent
\emph{State Transition:}
If the state at stage $k$ is $\psi$ (i.e., the source has already
forwarded the packet) then the next state is always $\psi$.
Suppose $(n,w,b)\in\mathcal{S}_k^a$  is
the state at some stage $k$, $0\le k\le K-1$, and $a_k\in\{0,1\}$
represents the action taken.
If $a_k=1$ then the decision process stops and we regard that the system enters the 
termination state $\psi$ so that the state at all the subsequent stages, $k+1,\cdots,K$, is $\psi$.
The source will also terminate the decision process, knowing that the relays wake-up within the interval $(0,T)$,
if it has waited for a duration of
$T$. This means that $(n,w,b)\in\mathcal{S}_k^a(2)$, i.e., $n=k-1$ and
$w=T$.

On the other hand if $(n,w,b)\in\mathcal{S}_k^a(1)$ and $a_k=0$, the source
waits for a random duration of $U_{k+1}$ and encounters a relay with a random
reward of $R_{k+1}$ so that the next state is $(n,w+U_{k+1},\max\{b,R_{k+1}\})$.
Note that if $n=k$, i.e., the current relay is the last one, then since we have
defined $U_{k+1}=T-w$ and $R_{k+1}=0$, the next state will be of the form $(k,T,b)$.
Thus the state at stage $k+1$ can be written down as,
\begin{equation}
\label{actual_state_transition_equn}
 s_{k+1}=\left\{\begin{array}{ll}
                 \psi \mbox{ if } w = T \mbox{ and/or } a_k=1\\
                 \Big(n,w+U_{k+1},\max\{b,R_{k+1}\}\Big)\mbox{ otherwise.}
                \end{array}\right.
\end{equation}
 
\subsection{Belief State and Belief State Transition}
Since the source does not know the actual number of relays $N$, the state is only partially
observable.  The source takes decisions based on the entire
history of the wake-up instants and the best rewards. If the source
has not forwarded the packet until stage $k-1$ then define,
$I_k=(p_0,(w_1,b_1),\cdots,(w_{k},b_{k}))$ to be the \emph{information vector}
available at the source when the $k$ th relay wakes up. $w_1,\cdots,w_k$
represents the wake-up instants of relays waking up at stages
$1,\cdots,k$ and $b_1,\cdots,b_k$ are the corresponding best rewards.
Define $p_k$ to be the \emph{belief state} about $N$ at stage $k$
given the information vector $I_k$, {i.e.,} $p_k(n)=\mathbb{P}(N=n|I_k)$ for $n=k,k+1,\cdots,K$
(note that $p_k(k)$ is the probability that the $k$ th relay is the last one).
\emph{Thus, $p_k$ is a pmf in the $K-k$ dimensional
probability simplex}. Let us denote this simplex as $\mathcal{P}_k$.
\begin{definition}
For $k=1,2,\cdots,K$, let $\mathcal{P}_k$:= 
set of all pmfs on the set $\{k,k+1,\cdots,K\}$. $\mathcal{P}_k$ is the $K-k$ dimensional
probability simplex in $\Re^{K-k+1}$.
\hfill $\blacksquare$
\end{definition}

The ``observation'' $(w_k,b_k)$ at stage $k$ is a part of the actual
state $(n,w_k,b_k)$. For a general POMDP problem the observation can
belong to a completely different space than the actual state space.
Moreover the distribution of the observation at any stage can in general
depend on all the previous states, observations, actions and disturbances.
Suppose this distribution depends only on the state, action and disturbance 
of the immediately preceding stage, then a belief on the actual state given 
the entire history turns out to be sufficient for taking decisions \cite[Chapter 5]{optimalcontrol}.
For our case, this condition is met and hence at stage $k$, $(p_k,w_k,b_k)$ is a 
\emph{sufficient statistic} to take decision.
Therefore we modify the state space as,
$\mathcal{S}_0=\{(p,0,0): p\in\mathcal{P}_1\}$ and for
$k=1,2\cdots,K$,
\begin{eqnarray}
\mathcal{S}_k=\Big\{(p,w,b): p\in\mathcal{P}_k, w\in(0,T], b\in[0,\overline{R}]\Big\}\cup\{\psi\}\mbox{.}
\end{eqnarray}

After seeing $k$ relays, suppose the source chooses not to forward the packet, then
upon the next relay waking up (if any), the source needs to update its belief
about the number of relays.
Formally, if $(p,w,b)\in\mathcal{S}_k$ is the state at stage $k$ and $w+u$ is
the wake-up instant of the next relay then, using Bayes rule, the next
belief state can be obtained via the following \emph{belief state
  transition function} which yields a pmf in $\mathcal{P}_{k+1}$,
\begin{eqnarray}
\label{belief_transition_equn}
 \tau_{k+1}(p,w,u)(n)=\frac{p(n)f_{k}(u|w,n)}{\sum_{\ell=k+1}^{K}p(\ell)f_{k}(u|w,\ell)}
\end{eqnarray}
for $n=k+1,\cdots,K$. Note that this function does not depend on $b$.
Thus, if at stage $k\in\{0,1,\cdots,K-1\}$, the state 
is $(p,w,b)\in\mathcal{S}_k$, then the next state is
\begin{equation}
 \label{state_transition_equn}
  s_{k+1}=\left\{\begin{array}{ll}
                 \psi \mbox{ if } w=T \mbox{ and/or } a_k=1\\
                 \Big(\tau_{k+1}(p,w,U_{k+1}),w+U_{k+1},\max\{b,R_{k+1}\}\Big)\mbox{ otherwise,}
                \end{array}\right.
\end{equation}
where $U_{k+1}$ is the random delay until the next relay wakes up and  $R_{k+1}$
is the random reward offered by that relay. The explanation for 
the above belief state transition expression remains same as that of the actual state transition
in (\ref{actual_state_transition_equn}), except that if the action is to continue, then 
the source needs to update the belief about the number of relays. 
Suppose at stage $k$, the actual number of relays happens to be $k$ and the action
is to continue, which is possible since the source does not know the
actual number, then the source will end up waiting until time $T$ and
then transmit to the relay with the best reward.  
\subsection{Stopping Rules and the Optimization Problem}
As the relays wake-up, the source's problem is to decide to stop or continue waiting
for further relays. A stopping rule or a  
policy $\pi$ is a sequence of mappings $(\mu_1,\cdots,\mu_{K})$
where $\mu_{k}:\mathcal{S}_k\rightarrow \{0,1\}$. 
Let $\Pi$ represent
the set of all policies.
The delay $D_{\pi}$ incurred using 
policy $\pi$ is the instant at which the source forwards 
the packet. It could be either one of the $W_k$, or the instant $T$.
The reward $R_{\pi}$ is the reward associated with the relay to which 
the packet is forwarded. The problem we are 
interested in is the following,
\begin{eqnarray}
 \label{main_problem}
 \min_{\pi\in\Pi}& \mathbb{E}[D_{\pi}]\nonumber\\
\mbox{Subject to}& \mathbb{E}[R_{\pi}]\ge\gamma\mbox{.} 
\end{eqnarray}
To solve the above problem, we consider the following unconstrained problem,
\begin{equation}
 \label{unconstrained_problem}
 \min_{\pi\in\Pi} \Big(\mathbb{E}[D_{\pi}]-\eta\mathbb{E}[R_{\pi}]\Big)
\end{equation}
where $\eta>0$. 
\begin{lemma}
 Let $\pi^{*}$ be an optimal policy for the unconstrained 
problem in (\ref{unconstrained_problem}). Suppose that $\eta$ $(=:\eta_\gamma)$ 
is such that $\mathbb{E}[R_{\pi^{*}}]=\gamma$, 
then $\pi^{*}$ is optimal for the main problem in (\ref{main_problem})
as well.
\end{lemma}
\begin{proof}
 For any policy $\pi$ satisfying the constraint 
$\mathbb{E}[R_{\pi}]\ge\gamma$ we can write,
\begin{eqnarray*}
\mathbb{E}[D_{\pi^{*}}]&\le&\mathbb{E}[D_{\pi}]-\eta_\gamma\Big(\mathbb{E}[R_{\pi}]-\mathbb{E}[R_{\pi^{*}}]\Big)\\
&=& \mathbb{E}[D_{\pi}]-\eta_\gamma\Big(\mathbb{E}[R_{\pi}]-\gamma\Big)\\
&\le& \mathbb{E}[D_{\pi}]\mbox{,}
\end{eqnarray*}
where the first inequality is by the optimality of 
$\pi^{*}$ for (\ref{unconstrained_problem}), the equality 
is by the hypothesis on $\eta_\gamma$, and the last 
inequality is due to the restriction  of $\pi$ to 
$\mathbb{E}[R_{\pi}]\ge\gamma$.
\end{proof}

\vspace{1.5 mm}
Hence we focus on solving the unconstrained problem 
in (\ref{unconstrained_problem}).
\subsection{One-Step Costs}
The objective in (\ref{unconstrained_problem}) can be seen as 
accumulating additively over each step. If the decision at a stage
is to continue then the delay until the next relay wakes up (or until $T$)
gets added to the cost. On the other hand if the decision is to stop then
the source collects the reward offered by the relay to which it forwards the
packet and the decision process enters the state $\psi$.
The cost in state $\psi$ is $0$.  Suppose $(p,w,b)$ is the state at stage $k$.  Then the one-step-cost
function is, for $k=0,1,\cdots,K-1$,
\begin{equation}
\label{cost_function_equn}
 g_k\Big((p,w,b),a_k\Big)=\left\{\begin{array}{ll}
                          -\eta b & \mbox{ if } w=T  \mbox{ and/or } a_k=1\\
                          U_{k+1} & \mbox{ otherwise. } 
                          \end{array}\right.
\end{equation}
The cost of termination is $g_{K}(p,w,b)=-\eta b$. 
Also note that for $k=0$, the possible states are of the form $(p,0,0)$ 
and the only possible action is $a_0=1$, so that  $g_0\Big((p,0,0),a_0\Big)=U_1$.
\subsection{Optimal Cost-to-go Functions}
For $k=1,2,\cdots,K$, let $J_k(\cdot)$ represent the \emph{optimal cost-to-go function} at stage $k$.
For any state $s_k\in\mathcal{S}_k$, $J_k(s_k)$ can be written as,
\begin{eqnarray}
\label{optimal_cost_prep_equn}
J_k(s_k)=\min\{\mbox{\emph{stopping cost}}, \mbox{\emph{continuing cost}}\}\mbox{,}
\end{eqnarray}
where \emph{stopping cost} (\emph{continuing cost}) represents the average cost incurred, if the source,
at the current stage decides to stop (continue), and takes optimal action at the subsequent stages.
For the termination state, since the one step cost is zero and since the system remains
in $\psi$ in all the subsequent stages, we have $J_k(\psi)=0$. 
For a state $(p,w,b)\in\mathcal{S}_k$, we next evaluate the two costs in the above expression.

First let us obtain the stopping cost. 
Suppose that there were $K$ relay nodes and the source has seen them all.  
In such a case if $(p,w,b)\in\mathcal{S}_{K}$ (note that $p$ will just be a point mass on $K$)
is the state at  stage $K$ then 
the optimal cost is simply the cost of termination, i.e., $J_{K}(p,w,b)=g_K(p,w,b)=-\eta b$.
For $k=1,2,\cdots,K-1$, if the action is to stop then the one step
 cost is $-\eta b$ and the next state is $\psi$ so that the further
 cost is $J_{k+1}(\psi)=0$.  Therefore, the 
 stopping cost at any stage is simply $-\eta b$. 

On the other hand the
cost for continuing, when the state at stage $k$ is $(p,w,b)$, using the total expectation law, can be written 
as, 
\begin{eqnarray}
\label{continue_cost_equn}
c_k(p,w,b)&=&p(k)\Big(T-w-\eta b\Big)\nonumber\\
&&+\sum_{n=k+1}^{K}p(n)\mathbb{E}_{k}\bigg[U_{k+1}+J_{k+1}\Big(\tau_{k+1}(p,w,U_{k+1}),
w+U_{k+1},\max\{b,R_{k+1}\}\Big)\bigg|w,n\bigg]\mbox{.}
\end{eqnarray}

Each of the expectation term in the summation in (\ref{continue_cost_equn}) is the  average cost to continue conditioned
on the event $(N=n)$.
$U_{k+1}$ is the (random) time until the next relay wakes up ($U_{k+1}$ is the one step cost) 
and  $J_{k+1}(\cdot)$ is 
the optimal cost-to-go from the next stage onwards ($J_{k+1}(\cdot)$ constitutes the future cost). 
The next state is obtained via the state transition equation (\ref{state_transition_equn}).
The term $(T-w-\eta b)$ in (\ref{continue_cost_equn}) associated with $p(k)$
is the cost of continuing when the number of relays happen to be $k$, i.e., 
$(N=k)$ and there are no more relays to go.
Recall that we had defined (in Section~\ref{system_model})
$U_{k+1}=T-w$ and $R_{k+1}=0$ when the actual number of relays is $N=k$.
Therefore $T-w$ is the one step cost when $N=k$.
Also $w+U_{k+1}=T$ and $\max\{b,R_{k+1}\}=b$ so that at the next stage
(which occurs at $T$) the process will terminate (enter $\psi$) with a
cost of $-\eta b$ (see (\ref{state_transition_equn}) and
(\ref{cost_function_equn})), which represents the future cost. 

Thus the optimal cost-to-go function (\ref{optimal_cost_prep_equn}) 
at stage $k=1,2,\cdots,K-1,$ can be written as,
\begin{eqnarray}
\label{optimal_cost_equn}
 J_k(p,w,b)=\min\Big\{-\eta b,c_k(p,w,b)\Big\}\mbox{.}
\end{eqnarray}
From the above expression it is clear that at stage $k$
when the state is $(p,w,b)$, the source 
has to compare the stopping cost, $-\eta b$,
with the cost of continuing, $c_k(p,w,b)$, and stop iff $-\eta b\le c_k(p,w,b)$.
Later in Section~\ref{bounds}, we will use this condition ($-\eta b\le c_k(p,w,b)$)
and define, the \emph{optimum stopping set}. We will prove that the continuing cost, $c_k(p,w,b)$, is 
concave in $p$, leading to the result that the optimum stopping set is convex.
(\ref{continue_cost_equn}) and (\ref{optimal_cost_equn}) are extensively used
in the subsequent development.

\section{Relationship with the Case Where $N$ is Known (the COMDP Version)}
\label{COMDP_section} 
In the previous section (Section~\ref{pomdp_formulation}) we detailed 
our problem formulation as a POMDP. 
The state is partially observable because the
source does not know the exact number of relays. It is interesting to first consider the
simpler case where this number is known,
which is the contribution of our earlier work in \cite{naveen-kumar10geographical-forwarding}.
Hence, in this section, we will consider the case when the initial pmf, $p_0$, has all
the mass only on some $n$, i.e., $p_0(n)=1$. We call this, the COMDP version of the problem. 

First we define a sequence of threshold 
functions which will be useful in the subsequent proofs. 
These are the same threshold functions that characterize the optimal policy for our model
in \cite{naveen-kumar10geographical-forwarding}.

\noindent
\begin{definition}
\label{phi_threshold_definition}
For $(w,b)\in(0,T)\times[0,\overline{R}]$, 
define $\{\phi_{\ell}:\ell=0,1,\cdots,K-1\}$ inductively as follows:
$\phi_0(w,b)=0$ for all $(w,b)$, and for $\ell=1,2,\cdots,K-1$ (recall Definition~\ref{Defn:expectation}),
 \begin{eqnarray}
\label{phi_k_equn}
 \phi_\ell(w,b)
&=&\mathbb{E}_{{K-\ell}}\bigg[\max\bigg\{b,R,\phi_{\ell-1}\Big(w+U,
\max\{b,R\}\Big)\bigg\}
-\frac{U}{\eta}\bigg|w,K\bigg]\mbox{.}
\end{eqnarray}
In the above expression we have suppressed the subscript ${K-\ell}+1$
for $R$ and $U$ for simplicity.  The pdf used to take the expectation
in the above expression is $f_{R}(\cdot)f_{K-\ell}(\cdot|w,K)$ (again recall Definition \ref{Defn:expectation}).  
\hfill $\blacksquare$
\end{definition}

\noindent
We will need the following simple property of the threshold functions in a later section.
\begin{lemma}
\label{phi_lemma_equn}
For $\ell=1,2,\cdots,K-1$, 
$-\eta\phi_\ell(w,b)\le(T-w-\eta b)$.
\end{lemma}
\begin{proof}
See Appendix~\ref{phi_lemma_equn_proof_appendix}.
\end{proof}

Next we state the main lemma of this section. We call this  the \emph{One-point} 
Lemma, because it gives the optimal cost, $J_k(p_k,w,b)$, at stage $k$ when the belief state
$p_k\in\mathcal{P}_k$ is such that it has all the mass on some $n\ge k$.
\begin{lemma}[One-point]
\label{one_point_mass_lemma}
Fix some $n\in\{1,2,\cdots,K\}$ and $(w,b)\in(0,T)\times[0,\overline{R}]$. 
For any $k=1,2,\cdots,n$, if  $p_k\in\mathcal{P}_k$ 
is such that $p_k(n)=1$
then,
\begin{equation*}
J_k(p_k,w,b)=\min\Big\{-\eta b,-\eta \phi_{n-k}(w,b)\Big\}\mbox{.}
\end{equation*}
\end{lemma}
\begin{proof}
The proof is by induction. 
We make use of the fact that if at some stage $k<n$ the belief state $p_k$ is
such that $p_k(n)=1$ then the next belief state $p_{k+1} (\in\mathcal{P}_{k+1})$, obtained by using the 
belief transition equation (\ref{belief_transition_equn}), is  also of the form $p_{k+1}(n)=1$.
We complete the proof by using Definition~\ref{phi_threshold_definition} and the induction hypothesis.
For a complete proof, see Appendix~\ref{one_point_mass_lemma_proof_appendix}.
\end{proof}

{\emph{Discussion of Lemma~\ref{one_point_mass_lemma}:}} 
At stage $k$ if the state is $(p_k,w,b)$,
where $p_k$ is such that $p_k(n)=1$ for some $n\ge k$, then from 
the One-point Lemma it follows that the optimal policy 
is to stop and transmit iff $b\ge\phi_{n-k}(w,b)$. The subscript 
${n-k}$ of the function $\phi_{n-k}$ signifies the number of 
more relays to go.
For instance, if we know that there are exactly 4 more relays 
to go then the threshold to be used is $\phi_4$. Suppose 
at stage $k$ if it was optimal to continue, then from 
(\ref{belief_transition_equn}) it follows that  the next 
belief state $p_{k+1}\in\mathcal{P}_{k+1}$ also has mass 
only on $(N=n)$   and hence at this stage it is optimal to 
use the threshold function $\phi_{n-(k+1)}$. Therefore, if 
we begin with
an intial belief $p_0\in\mathcal{P}_1$ such that $p_0(n)=1$ 
for some $n$, then the optimal policy is to stop at the first 
stage $k$ such that $b\ge\phi_{n-k}(w,b)$ 
where $W_k=w$ is the wake-up instant of the $k$ th 
relay and $B_k=\max\{R_1,\cdots,R_k\}=b$. Note that, since 
at stage $n$ the threshold to be used is $\phi_0(w,b)=0$ (see Definition~\ref{phi_threshold_definition}), 
we invariably have to stop at stage $n$ if we have 
not terminated earlier. This is exactly the same as our optimal policy in 
\cite{naveen-kumar10geographical-forwarding},
where the number of relays is known to the source 
(instead of knowing the number wp1, as in our One-point Lemma here).
\hfill$\blacksquare$
\section{Unknown $N$: Bounds on the optimum stopping set}
\label{bounds}
In this section we will consider the general case where the number of relays $N$ is not
known to the source. The sequential decision problem developed in Section~\ref{pomdp_formulation}
was for this unknown $N$ case. The problem was formulated as a 
POMDP for which the source's
decision to stop and forward the packet is based on the belief state which 
takes values in $\mathcal{P}_k$ after the source has observed $k$ relays
waking up. We begin this section by defining the \emph{optimum stopping set}. 
We show that this set is convex. Characterizing the exact optimum stopping set 
is computationally intensive. Therefore, our aim is to derive
inner and outer bounds (a subset and a superset, respectively) for the optimum stopping set.

\begin{definition}[Optimum stopping set]
\label{optimal_stopping_set_defn}
For $1\le k\le K-1$, let $\mathcal{C}_k(w,b)=\Big\{p\in\mathcal{P}_k: -\eta b\le c_k(p,w,b)\Big\}$. 
Referring to (\ref{optimal_cost_equn}) it follows that, for a given $(w,b)$, $\mathcal{C}_k(w,b)$
represents the set of all beliefs $p\in\mathcal{P}_k$ at stage $k$ at which 
it is optimal to stop. We call $\mathcal{C}_k(w,b)$
the \emph{optimum stopping set} at stage $k$
when the delay ($W_k$) and best reward ($B_k$) values are $w$ and $b$, respectively. 
\hfill $\blacksquare$
\end{definition}

\subsection{Convexity of the Optimum Stopping Sets}
We will prove (in Lemma~\ref{concave_lemma}) that the continuing cost, $c_k(p,w,b)$, in (\ref{continue_cost_equn}) is 
concave in $p\in\mathcal{P}_k$. From the form of the stopping set $\mathcal{C}_k(w,b)$, 
a simple consequence of this lemma will be that the optimum stopping set 
is convex. We further extend the concavity result of $c_k(p,w,b)$ for $p\in{\overline{\mathcal{P}}}_k$,
where ${\overline{\mathcal{P}}}_k$ is the affine set containing $\mathcal{P}_k$ (to be defined shortly in this
section).

\label{convex_section}
\begin{lemma}
\label{concave_lemma}
For $k=1,2,\cdots,K-1,$ and any given $(w,b)$, the cost of continuing (defined in (\ref{continue_cost_equn})),
$c_k(\cdot,w,b)$, is concave on $\mathcal{P}_k$. 
\end{lemma}
\begin{proof}
The essence of the proof is same as that in \cite[Lemma~1]{porta-etal06point-based}.
From (\ref{continue_cost_equn}) we easily see that $c_{K-1}(\cdot,w,b)$ is an affine function of $p\in\mathcal{P}_{K-1}$,
and hence $J_{K-1}(\cdot,w,b)$, in (\ref{optimal_cost_equn}), being minimum of an affine function and a constant is concave.
The proof then follows by induction.
The induction hypothesis is that for some stage $k+1$, $J_{k+1}(\cdot,w,b)$ is concave. Hence it
can be expressed as an infimum over some collection of affine functions.
The inductive step then shows that
$c_k(\cdot,w,b)$ can also be similarly expressed as
an infimum over some collection of affine functions. Hence $c_k(\cdot,w,b)$ and 
(using \ref{optimal_cost_equn}) $J_k(\cdot,w,b)$ are concave.
Formal proof is available in Appendix~\ref{concave_proof_appendix}.
\end{proof}

The following corollary is a straight forward application 
of the above lemma.
\begin{corollary}
\label{convex_corollary}
For $k=1,2,\cdots,K-1,$ and any given $(w,b)$, $\mathcal{C}_k(w,b)(\subseteq\mathcal{P}_k)$ is a convex set.
\end{corollary}
\begin{proof}
 From Lemma~\ref{concave_lemma} we know 
that $c_k(p,w,b)$ is a concave function of $p\in\mathcal{P}_k$. 
 Hence $\mathcal{C}_k(w,b)$ (see Definition~\ref{optimal_stopping_set_defn}), being a 
\emph{super level set} of a concave function, 
is convex \cite{convexoptimization}.
\end{proof}

In the next section while proving an inner bound for the stopping set
$\mathcal{C}_k(w,b)$, we will identify a set of points that could lie
outside the probability simplex $\mathcal{P}_k$.  We can obtain a
better inner bound if we extend the concavity result to the affine
set,
\begin{eqnarray*}
 {\overline{\mathcal{P}}}_k&=&\Big\{p\in\Re^{K-k+1}:\left<p,\textbf{1}\right>=1\Big\}\mbox{,}
\end{eqnarray*}
where $\left<p,\textbf{1}\right>=\sum_{n=k}^{K}p(n)$, i.e., in ${\overline{\mathcal{P}}}_k$ 
the vectors sum to one, but
we do not require non-negativity of the vectors.  This can be done as
follows.  Define ${\overline{\tau}}_{k+1}(p,w,u)$ using
(\ref{belief_transition_equn}) for every
$p\in\overline{\mathcal{P}}_k$.  Then ${\overline{\tau}}_{k+1}(.,w,u)$
as a function of $p$, is the extension of $\tau_{k+1}(.,w,u)$ from
$\mathcal{P}_k$ to $\overline{\mathcal{P}}_k$. Similarly, for every
$p\in\overline{\mathcal{P}}_k$, define ${\overline{c}}_k(p,w,b)$ and
${\overline{J}}_k(p,w,b)$ using (\ref{continue_cost_equn}) and
(\ref{optimal_cost_equn}). These are the extensions of $c_k(\cdot,w,b)$
and $J_k(\cdot,w,b)$ respectively. Then again, using the proof technique
same as that in Lemma~\ref{concave_lemma}, we can
obtain the following corollary,

\begin{corollary}
\label{concave_affine_lemma}
For $k=1,2,\cdots,K-1$, and any given $(w,b)$,
 ${\overline{c}}_k(\cdot,w,b)$ is concave on 
the affine set  
${\overline{\mathcal{P}}}_k$. 
\hfill $\blacksquare$
\end{corollary}

\noindent
Using the above corollary,
 $\mathcal{C}_k(w,b)$ can be written as,
\begin{eqnarray}
 \label{stopping_set_equn}
 \mathcal{C}_k(w,b)&=& \mathcal{P}_k \cap
 \Big\{p\in\Re^{K-k+1}:\left<p,\textbf{1}\right>=1,
  -\eta b\le {\overline{c}}_k(p,w,b)\Big\} \mbox{.} 
\end{eqnarray}

\subsection{Inner Bound on the Optimum Stopping Set}
\label{inner_bound}
We have showed that the optimum stopping set is convex. In this section, we will identify 
points that lie along certain edges of the simplex $\mathcal{P}_k$. A convex hull of these
points will yield an inner bound to the optimum stopping set.
This will first require us to prove the following lemma, referred to as the 
\emph{Two-points} Lemma, and is a generalization of the One-point Lemma (Lemma~\ref{one_point_mass_lemma}).
It gives the optimal cost, $J_k(p,w,b)$, at stage $k$ when $p\in\mathcal{P}_k$ is such that
it places all its mass on $k$ and on some $n>k$, i.e., $p(k)+p(n)=1$. 
Throughout this and the next section
(on an outer bound) $(W_k,B_k)=(w,b)$ is fixed and hence, for the ease of presentation (and readability), we drop 
$(w,b)$ from the notations $\delta_{\ell}(w,b)$, $a_k^{\ell}(w,b)$ and 
$b_k^{\ell}(w,b)$ (to appear in these sections later). However it is understood that
these thresholds are, in general, functions of $(w,b)$.
\begin{lemma}[Two-points]
\label{two_point_mass_lemma}
 For $k=1,2,\cdots,K-1,$ if  
$p\in\mathcal{P}_k$ is such that $p(k)+p(n)=1$,
where  $k < n\le K$ then,
\begin{eqnarray*}
{J_k(p,w,b)}
&=&\min\Big\{-\eta b, p(k)\Big(T-w-\eta b\Big) + 
p(n)\Big(-\eta \phi_{n-k}(w,b)\Big)\Big\}\mbox{.}
\end{eqnarray*}
\end{lemma}
\begin{proof}
Using (\ref{continue_cost_equn}) we can write,
\begin{eqnarray*}
c_k(p,w,b)&=&p(k)\Big(T-w-\eta b\Big)\\
&&+p(n)\mathbb{E}_{k}\bigg[U_{k+1}+
J_{k+1}\Big(\tau_{k+1}(p,w,U_{k+1}),w+U_{k+1},\max\{b,R_{k+1}\}\Big)\bigg|w,n\bigg]\mbox{.}
\end{eqnarray*}
For $p$ given as in the hypothesis, the belief in the next state is
such that $\tau_{k+1}(p,w,u)(n)=1$.  Using this observation, 
Lemma~\ref{one_point_mass_lemma} (One-point), and the definition of $\phi_{n-k}$ in (\ref{phi_k_equn}),
we obtain the desired result.
\end{proof}

\emph{Discussion of Lemma~\ref{two_point_mass_lemma}}: The Two-points Lemma
(Lemma~\ref{two_point_mass_lemma}) can be used to obtain certain threshold points
in the following way.
When $p\in\mathcal{P}_k$ has mass only on $k$ and on some $n$, $k<n\le K$, then
using Lemma~\ref{two_point_mass_lemma}, the continuing cost can be
written as a function of $p(n)$ as,
\begin{eqnarray}
\label{ck_threshold_equn}
c_k(p,w,b)&=&\Big(T-w-\eta b\Big)-
p(n)\Big(T-w-\eta \Big(b-\phi_{n-k}(w,b)\Big)\Big)\mbox{.}
\end{eqnarray}
From Lemma \ref{phi_lemma_equn}, it follows that $c_k(p,w,b)$ in  
(\ref{ck_threshold_equn})  is a
decreasing function of $p(n)$. Let $p_k^{(k)}$ and $p_k^{(n)}$ be {pmfs} in 
$\mathcal{P}_k$ with mass only on $N=k$ and $N=n$ respectively.
These are two of the corner points of the simplex $\mathcal{P}_k$ 
(as an example, Fig.~\ref{simplex} illustrates the simplex and the corner
points for stage $k=K-2$. With at most two more nodes to go, $\mathcal{P}_{K-2}$
is a two dimensional simplex in $\Re^3$. $p_{K-2}^{(K-2)}$, 
$p_{K-2}^{(K-1)}$ and $p_{K-2}^{(K)}$ are the corner points of this simplex).
\begin{figure}
\centering
\includegraphics[scale=0.35]{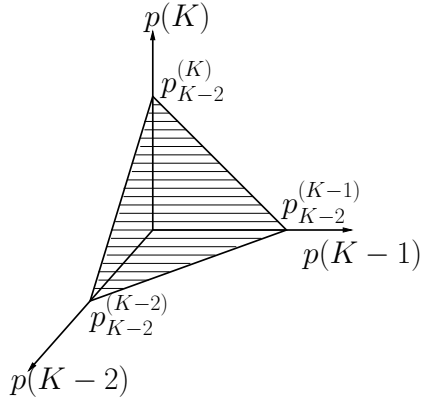}
\caption{Probability simplex, $\mathcal{P}_{K-2}$, at stage $K-2$. A belief state at stage $K-2$ is a pmf on the points
$K-2$, $K-1$ and $K$ (i.e., no-more, one-more and two-more relays to go, respectively).
Thus $\mathcal{P}_{K-2}$ is a two dimensional simplex in $\Re^3$.\label{simplex}}
\end{figure}
At stage $k$  as we move along the line joining the points 
$p_k^{(k)}$ and $p_k^{(n)}$ (Fig.~\ref{cost_case1} and \ref{cost_case2}
illustrates this as $p(n)$ going from $0$ to $1$), the cost of continuing in (\ref{ck_threshold_equn}) decreases 
and there is a threshold below which it is optimal to 
transmit and beyond which it is optimal to continue. 
The value of this threshold is that value of $p(n)$ in 
(\ref{ck_threshold_equn}) at which the  continuing cost 
becomes equal to $-\eta b$. 
Let $\delta_{n-k}$ 
denote this threshold value, then
\begin{eqnarray*}
 \delta_{n-k}&=& \frac{T-w}{T-w-\eta \Big(b-\phi_{n-k}(w,b)\Big)}\mbox{.}
\end{eqnarray*}
The cost of continuing in (\ref{ck_threshold_equn}) as a function of $p(n)$
along with the stopping cost, $-\eta b$, is shown in Fig.~\ref{cost_case1} and \ref{cost_case2}.
The threshold $\delta_{n-k}$ is the point of intersection of these two cost functions.
The value of the continuing cost $c_k(p,w,b)$ at $p(n)=1$ is $-\eta \phi_{n-k}(w,b)$.
Note that in the case when $b>\phi_{n-k}(w,b)$ the 
threshold $\delta_{n-k}$ will be greater than $1$ 
in which case it is optimal to stop for any $p$ on 
the line joining $p_k^{(k)}$ and $p_k^{(n)}$. 
\hfill $\blacksquare$
\begin{figure}[h]
\centering
\subfigure[]{
\includegraphics[scale=0.35]{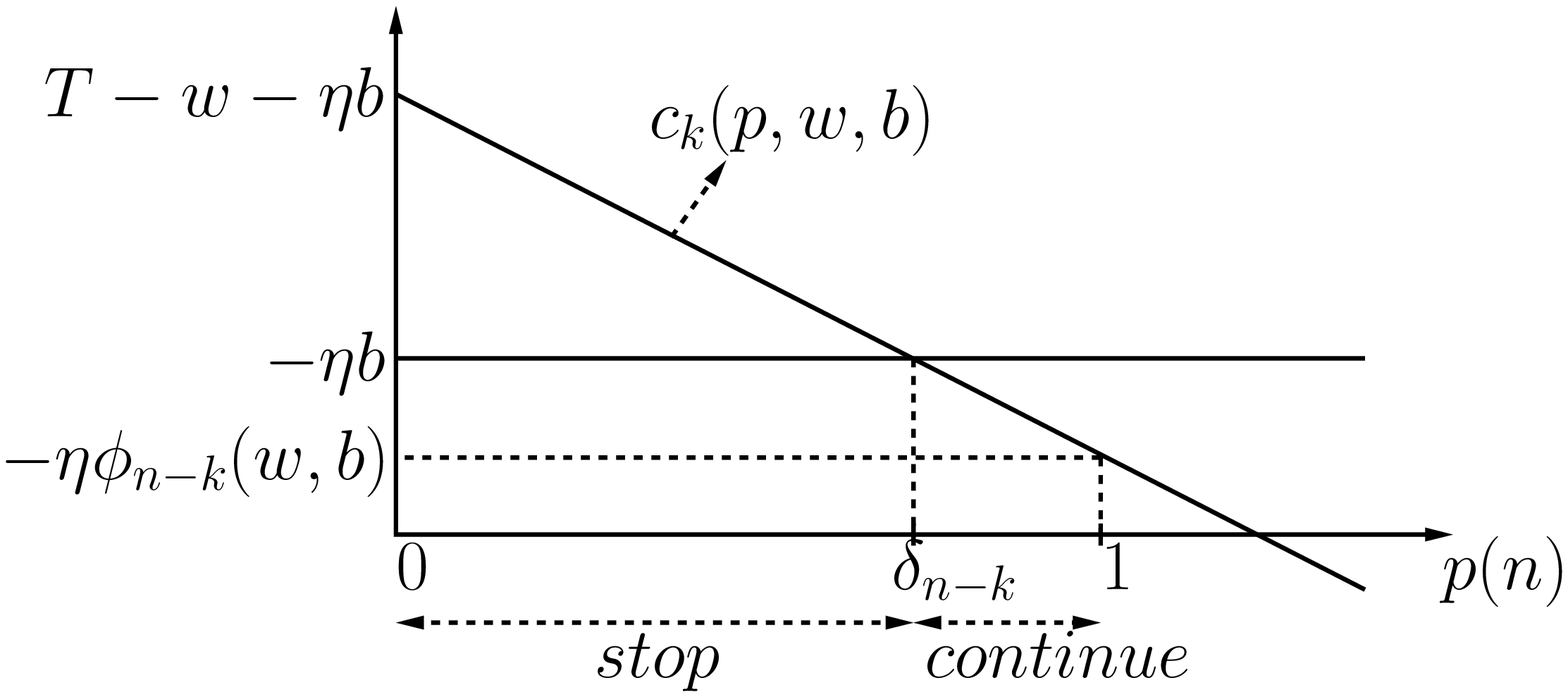}
\label{cost_case1}
}
\subfigure[]{
\includegraphics[scale=0.35]{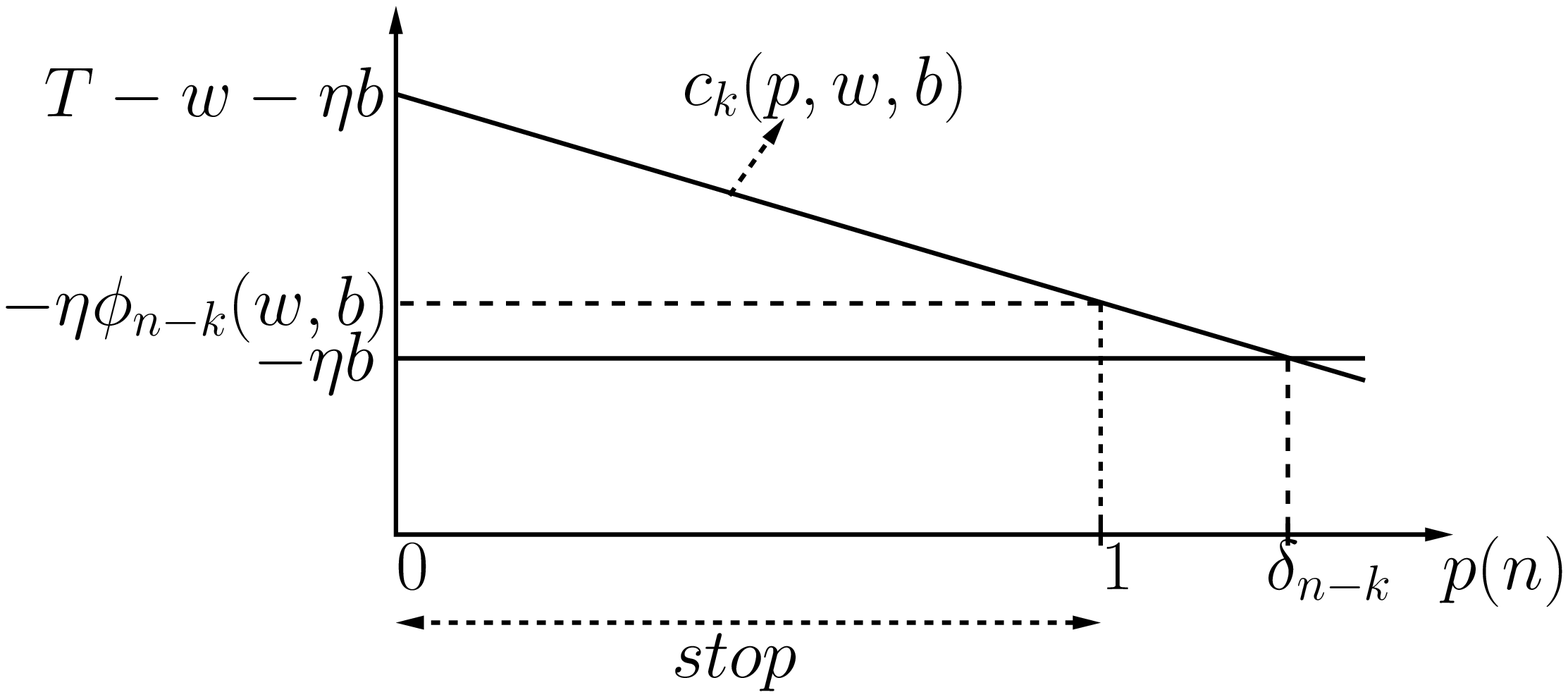}
\label{cost_case2}
}
\caption{\label{cost_cases_figure} Depiction of the thresholds
  $\delta_{n-k}(w,b)$.  $c_k(p,w,b)$ in
  Equation~(\ref{ck_threshold_equn}) is plotted as a function of
  $p(n)$.  
  Also shown is the constant function $-\eta b$ which is the stopping cost.
  $\delta_{n-k}$ is the point of intersection of
  these two functions.  \subref{cost_case1} When $b\le
  \phi_{n-k}$.  \subref{cost_case2} When $b>\phi_{n-k}(w,b)$.}
\end{figure}

There are similar thresholds along each edge of the simplex 
$\mathcal{P}_k$ starting from the corner point $p_k^{(k)}$. In general, 
let us define for $\ell=1,2,\cdots,K$,
\begin{eqnarray}
\label{a_l_equn}
 \delta_{\ell}&=& \frac{T-w}{T-w-\eta \Big(b-\phi_{\ell}(w,b)\Big)}\mbox{.}
\end{eqnarray}

\emph{Remark:} Note that (\ref{ck_threshold_equn}) will also hold for the extended function 
$\overline{c}_k(p,w,b)$, where now $p\in\overline{\mathcal{P}}_k$.
In terms of the extended function, $\delta_{n-k}$ represents the value of $p(n)$ 
(in (\ref{ck_threshold_equn}) with $c_k$ replaced by $\overline{c}_k$) at which $\overline{c}_k(p,w,b)=-\eta b$.

Recall that (from Lemma~\ref{two_point_mass_lemma}) the above discussion began with a $p\in\mathcal{P}_k$ 
such that $p(k)+p(n)=1$. At the threshold of interest we have $p(n)=\delta_{n-k}$ and hence
$p(k)=1-\delta_{n-k}$, and the rest of the components are zero. We denote this vector 
as $a_k^{n-k}$. For instance in Fig.~\ref{lower_bound_figure}, where the 
face of the two dimensional simplex $\mathcal{P}_{K-2}$ is shown, the threshold along the lower
edge of the simplex is $a_{K-2}^{1}=[1-\delta_1,\delta_1,0]$
and that along the other edge is $a_{K-2}^{2}=[1-\delta_2,0,\delta_2]$.
Since it is possible for $\delta_{n-k}>1$, therefore the vector threshold $a_k^{n-k}$
is not restricted to lie in the simplex $\mathcal{P}_k$, however it always stays in the affine set ${\overline{\mathcal{P}}}_k$.
We formally define these thresholds next.

\begin{figure*}
\centering
\subfigure[]{
\includegraphics[scale=0.28]{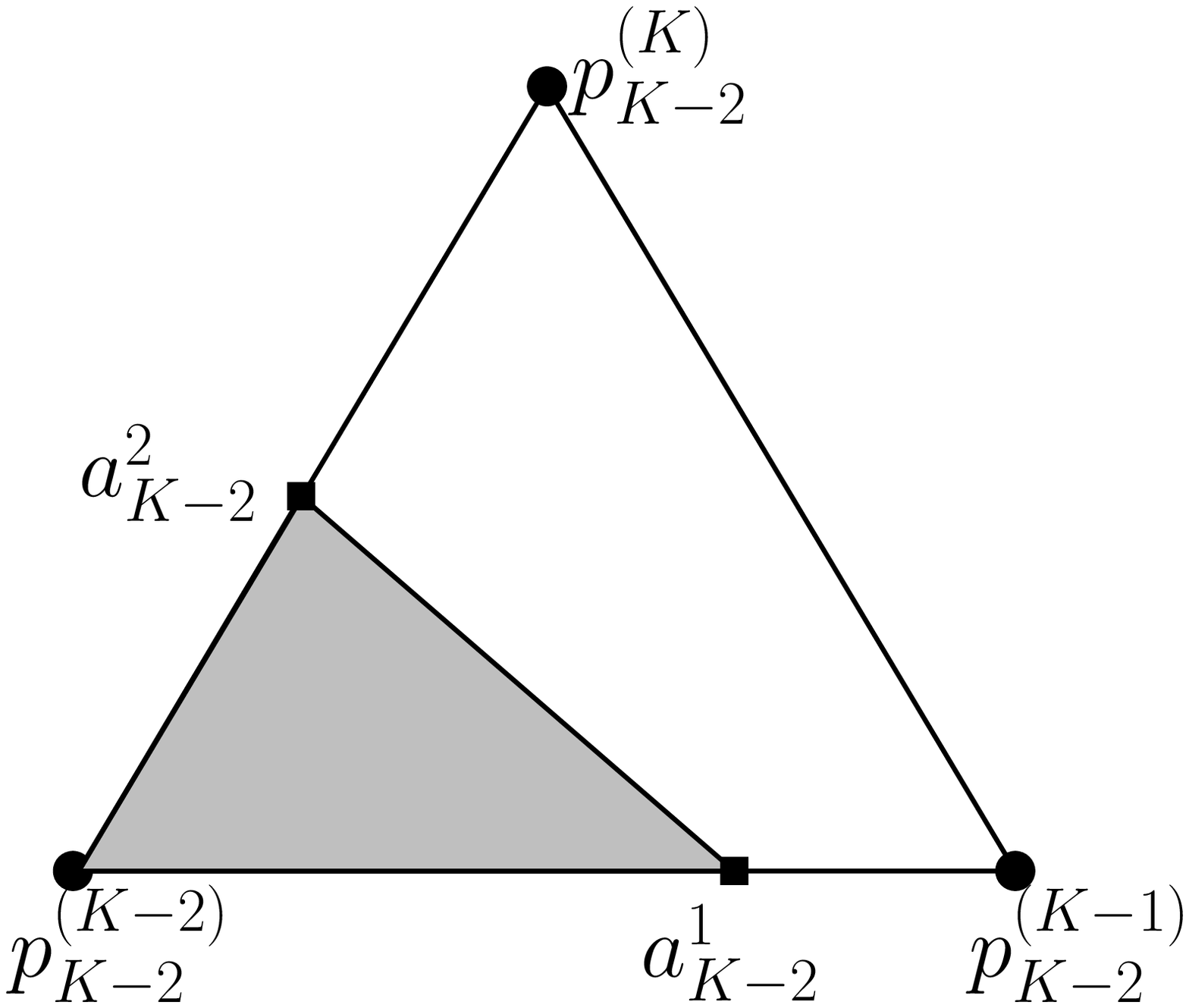}
\label{lower_case1}
}
\centering
\subfigure[]{
\includegraphics[scale=0.28]{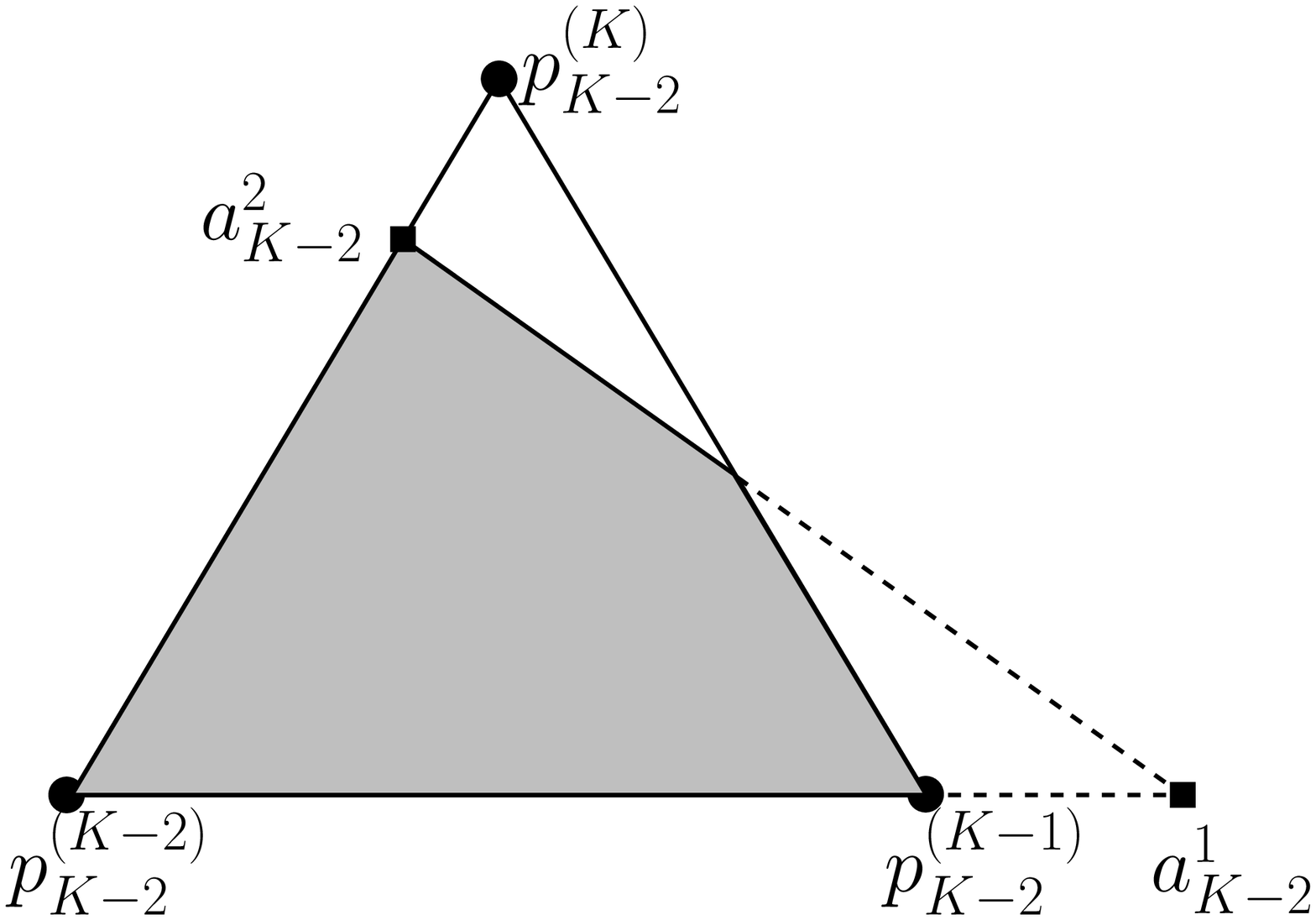}
\label{lower_case2}
}
\centering
\subfigure[]{
\includegraphics[scale=0.28]{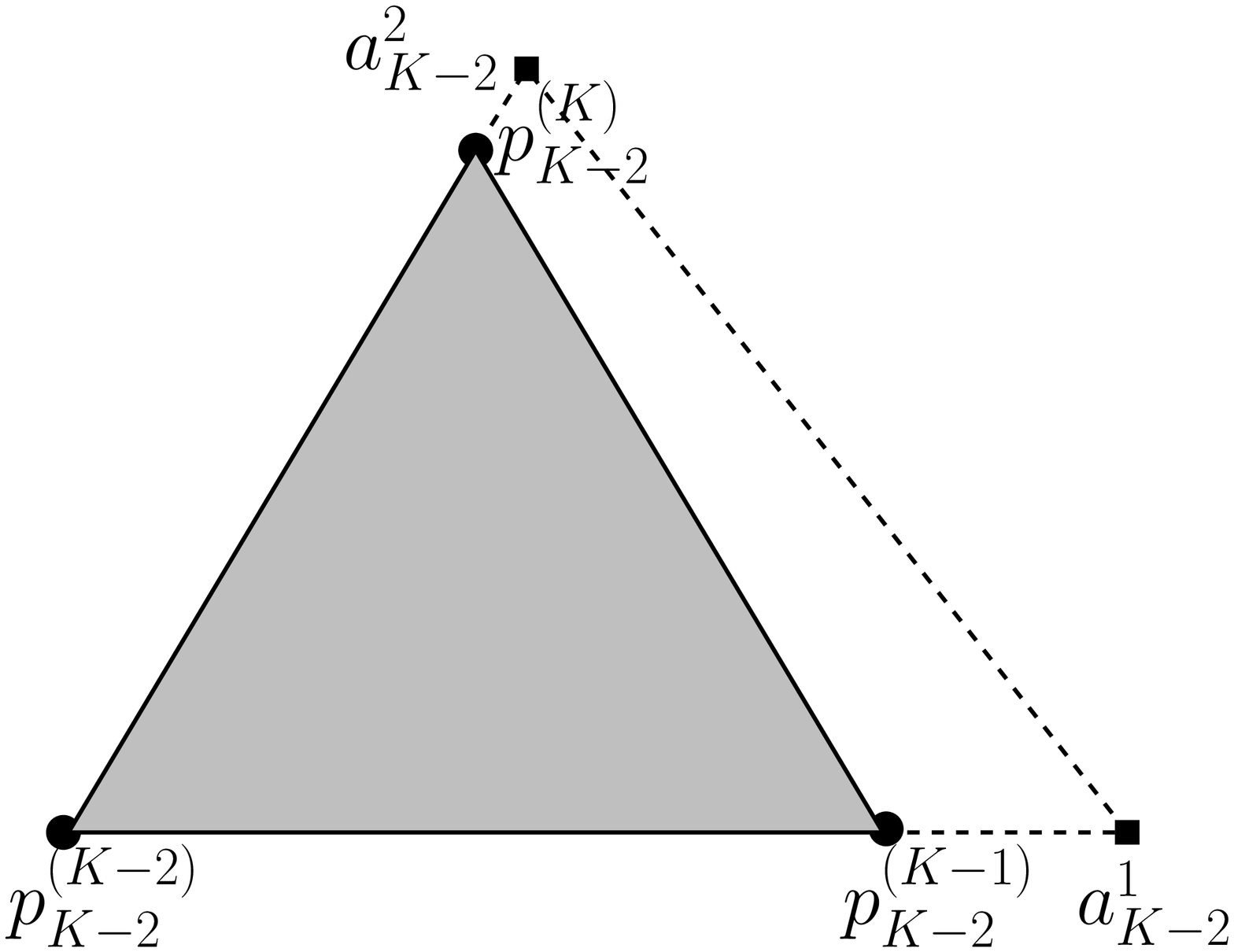}
\label{lower_case3}
}
\caption{\label{lower_bound_figure} Depiction of the inner bound
  $\underline{\mathcal{C}}_{K-2}(w,b)$. 
  In the examples in \subref{lower_case1},
  \subref{lower_case2}, and \subref{lower_case3} we only show the face
  of the simplex, $\mathcal{P}_{K-2}$, in Fig.~\ref{simplex}, with the inner bound being shown as the shaded
  region.  \subref{lower_case1} When $\delta_{1}$ and
  $\delta_2$ are both less than $1$.  \subref{lower_case2}
  When $\delta_{1}>1$ and $\delta_{2}<1$.
  \subref{lower_case3} When $\delta_{1} > 1$ and
  $\delta_2 > 1$.}
\vspace*{-4 mm}
\end{figure*}

\begin{definition}
For a given $k\in\{1,2,\cdots,K-1\}$, for each  $\ell=1,2,\cdots,K-k$ define  $a_k^{\ell}$ as a 
$K-k+1$ dimensional point with the first and the $\ell+1$ th 
components equal to   $1-\delta_\ell$  and  
$\delta_\ell$ respectively, the rest of the components are zeros.
As mentioned before, $a_k^{\ell}$ lies on the line joining $p_k^{(k)}$
and $p_k^{(k+\ell)}$. At stage $k$ there are $K-k$ such points, one corresponding to
each edge in $\mathcal{P}_k$ emanating from
the corner point $p_k^{(k)}$.
For an illustration of these points see Fig.~\ref{lower_bound_figure} for the case $k=K-2$.
\hfill $\blacksquare$
\end{definition}

Referring to Fig.~\ref{lower_case1} (which depicts the case, $k=K-2$), 
suppose all the vector thresholds,
$a_k^{l}$, lie within the simplex $\mathcal{P}_k$ then,
since at these points the stopping cost $(-\eta b)$ is equal to the continuing cost
($c_k(a_k^{l},w,b)$), all these points lie in the optimum stopping set $\mathcal{C}_k(w,b)$.
Note that  the corner point $p_k^{(k)}$ (belief with all the mass on no-more relays to go)
also lies in $\mathcal{C}_k(w,b)$.
Since we have already shown that $\mathcal{C}_k(w,b)$ is convex, 
the convex hull of these points will yield an inner bound.
However as mentioned earlier (and as depicted in Fig.~\ref{lower_case2} and 
\ref{lower_case3})  it is possible for some or all the thresholds
$a_k^{l}$ to lie outside the simplex (and hence these thresholds do not 
belong to $\mathcal{C}_k(w,b)$).
This is where
we will use Corollary~\ref{concave_affine_lemma}, where the concavity result
of the continuing cost, $c_k(p,w,b)$, is extended to the affine set ${\overline{\mathcal{P}}}_k$.
We next state this \emph{inner bound} theorem:

\begin{theorem}[Inner bound]
\label{inner_bound_theorem}
 For $k=1,2,\cdots,K-1$, Recalling that $p_k^{(k)}$ is the pmf in $\mathcal{P}_k$ with point mass on $k$, define
\begin{eqnarray*} 
\underline{\mathcal{C}}_k(w,b):=\mathcal{P}_k \cap conv\Big\{p_k^{(k)},a_k^1,\cdots,a_k^{K-k}\Big\}\mbox{,}
\end{eqnarray*}
where $conv$ denotes the convex hull of the given 
points. Then $\underline{\mathcal{C}}_k(w,b)\subseteq\mathcal{C}_k(w,b)$.
\end{theorem}
\begin{proof}
The way the points $a_k^{\ell}$ are defined using $\delta_{\ell}$ it follows that 
${\overline{c}}_k(a_k^{\ell},w,b)=-\eta b$ (see Remark following
(\ref{a_l_equn})). $p_k^{(k)}$ is the pmf with point mass on $(N=k)$, so that 
$\overline{c}_k(p_k^{(k)},w,b)=c_k(p_k^{(k)},w,b)=T-w-\eta b$ (see (\ref{continue_cost_equn})). Therefore the
 points $p_k^{(k)},a_k^1,\cdots,a_k^{K-k}\in\Big\{p\in\Re^{K-k+1}:p.\textbf{1}=1, -\eta b\le {\overline{c}}_k(p,w,b)\Big\}$ 
which is a convex set (because ${\overline{c}}_k(p,w,b)$ is concave in $p$, from Corollary~\ref{concave_affine_lemma}). 
Therefore
\begin{eqnarray*}
 conv\Big\{p_k^{(k)},a_k^1,\cdots,a_k^{K-k}\Big\}&\subseteq&
\Big\{p\in\Re^{K-k+1}:p.\textbf{1}=1, 
-\eta b\le {\overline{c}}_k(p,w,b)\Big\}
\end{eqnarray*}
and the result follows from (\ref{stopping_set_equn}).
\end{proof}

In Fig.~\ref{lower_bound_figure}, for stage $k=K-2$, we illustrate the various cases that can arise.
In each of the figures
the shaded region is the inner bound. In Fig.~\ref{lower_case1} all the thresholds lie within the simplex
and simply the convex hull of these points gives the inner bound. When some or all the thresholds
lie outside the simplex, as in Fig.~\ref{lower_case2} and \ref{lower_case3}, then the inner bound is 
obtained by intersecting the convex hull of the thresholds with the simplex. 
In Fig.~\ref{lower_case3}, where all the thresholds lie outside the simplex, the inner bound
is the entire simplex, $\mathcal{P}_{K-2}$, so that at stage $K-2$ with $(W_{K-2},B_{K-2})=(w,b)$ 
it is optimal to stop for any belief state.

\subsection{Outer Bound on the Optimum Stopping Set}
\label{outer_bound}
In this section we will obtain an outer bound (a superset) for 
the optimum stopping set. Again, as in the case of the inner bound,
we will identify certain threshold points whose convex hull will
contain the optimum stopping set.
This will require us to first prove a monotonicity result which
compares the cost of continuing  at two  belief states $p,q\in\mathcal{P}_k$
which are ordered, for instance for $k=K-2$, as in Fig~\ref{upper_bound_figure}.
$q$ in Fig.~\ref{upper_bound_figure} is such that $q(K-2)=p(K-2)$ (i.e.,
the probability that there is no-more relays to go is same in both $p$ and $q$)
and $q(K-1)=1-p(K-2)$ (i.e., all the remaining probability in $q$ is on the event
that there is one-more relay to go, while in $p$ it can be on one-more or two-more relays to go).
Thus $q$ lies on the lower edge of the simplex.
We will show that the cost of continuing at $p$ is less than that at $q$.

\begin{lemma}
\label{cost_order_lemma}
 Given $p\in\mathcal{P}_{k}$ for $k=1,2,\cdots,K-1$, define 
$q(k)=p(k)$ and $q(k+1)=1-p(k)$, then $c_{k}(p,w,b)\le c_{k}(q,w,b)$
 for any $(w,b)$.
\end{lemma}
\begin{proof}
See Appendix~\ref{cost_order_proof_appendix}. 
\end{proof}

\emph{Discussion of Lemma~\ref{cost_order_lemma}:}
This lemma proves the intuitive result
that the continuing cost with a pmf $p$ that gives mass on a larger
number of relays should be smaller than with a pmf $q$ that
concentrates all such mass in $p$ on just one more relay to go. With
more relays, the cost of continuing is expected to decrease.
\hfill $\blacksquare$

Similar to the thresholds $a_k^\ell$ we define the thresholds $b_k^\ell$ that lie along certain
edges of the simplex. We will identify the threshold $a_k^{\ell}$
that is at a maximum distance from the corner point $p_k^{(k)}$ (in Fig.~\ref{upper_bound_figure}, this point
is $a_{K-2}^{1}=[1-\delta_{1},\delta_{1},0]$). Next we define the thresholds $b_{k}^{\ell}$
to be the points on the edges emanating from $p_{k}^{(k)}$, which are at this same distance.
Thus in Fig.~\ref{upper_bound_figure}, $b_{K-2}^{1}=a_{K-2}^{1}$ and $b_{K-2}^{2}=[1-\delta_{1},0,\delta_{1}]$.

\begin{figure}[h]
 \centering
 \includegraphics[scale=0.28]{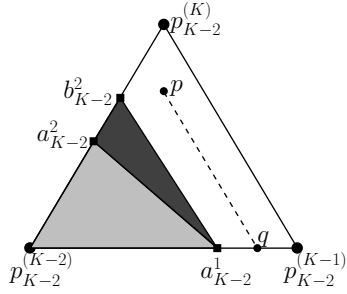}
\caption{\label{upper_bound_figure} 
The light shaded region is the inner bound. 
The outer bound is the union of 
 the light and the dark shaded regions.}
\end{figure}

\noindent
\begin{definition}
Now for $\ell=1,2,\cdots,K-k$ define
$b_k^{\ell}$ as a $K-k+1$ dimensional point with the first and
the $\ell+1$ th components equal to
$1-\delta_{\ell_{max}}$ and
$\delta_{\ell_{max}}$ respectively, the rest
of the components are zeros. Each of the $b_k^{\ell}$ 
are at equal distance from $p_k^{(k)}$ but on a different edge starting from $p_k^{(k)}$.
\hfill $\blacksquare$
\end{definition}
\noindent

Using Lemma~\ref{cost_order_lemma}, we show that the convex hull 
of the thresholds $b_k^{l}$ along with the corner point $p_k^{(k)}$
constitutes an outer bound for the optimum stopping set.
The idea of the proof can be illustrated using Fig.~\ref{upper_bound_figure}.
$p$ in Fig.~\ref{upper_bound_figure} is outside the convex hull
and $q$ is obtained from $p$ as in Lemma~\ref{cost_order_lemma}.
At $q$ it is optimal to continue since it is beyond the threshold $a_{K-2}^1$
and hence the continuing cost at $q$, $c_k(q,w,b)$, is less than the stopping cost 
$-\eta b$. From Lemma~\ref{cost_order_lemma} it follows that the continuing cost 
at $p$, $c_k(p,w,b)$, is also less than $-\eta b$ so that it is optimal
to continue at $p$ as well, proving that $p$ does not belong to the optimum stopping set.
Thus the convex hull contains the optimum stopping set.
We formally state and prove this \emph{outer bound} theorem next.
\begin{theorem}[Outer bound]
\label{outer_bound_theorem}
For $k=1,2,...,K-1,$ define
\begin{eqnarray*}
\overline{\mathcal{C}}_k(w,b)=\mathcal{P}_k \cap conv\Big\{p_k^{(k)},b_k^1,\cdots,b_k^{K-k}\Big\}\mbox{.}
\end{eqnarray*}
Then $\mathcal{C}_{k}(w,b)\subseteq \overline{\mathcal{C}}_{k}(w,b)$.
\end{theorem}
\begin{proof}
  Let $\ell_{max}=\argmax_{\ell=1,2,\cdots,K-k}\delta_{\ell}$. If
  $\delta_{\ell_{max}}\ge 1$, then $\overline{\mathcal{C}}_{k}
  (w,b)=\mathcal{P}_{k} (\supseteq \mathcal{C}_{k}(w,b))$ and the
  result trivially follows.  Hence, let us consider the case where $\delta_{\ell_{max}}<1$. 
  Pick any $p\notin \overline{\mathcal{C}}_{k}(w,b)$. We will show that $p\notin\mathcal{C}_{k}(w,b)$.
  Let $q\in\mathcal{P}_{k}$ be such that $q(k)=p(k)$ and
  $q(k+1)=1-p(k)$.

  $p\notin \overline{\mathcal{C}}_k(w,b)$ implies that
  $p(k)<1-\delta_{\ell_{max}}$. Since $q(k+1)=1-p(k)>\delta_{\ell_{max}}\ge \delta_1$, it follows that under $q$ it is
  optimal to continue so that $q\notin \mathcal{C}_{k}(w,b)$ i.e., $c_k(q,w,b)<-\eta b$.
  Finally by applying Lemma~\ref{cost_order_lemma} we can write $c_{k}(p,w,b)\le c_{k}(q,w,b)<-\eta b$. 
  This means that at $p$ it is optimal to continue so that $p\notin\mathcal{C}_{k}(w,b)$.
\end{proof}

The outer bound for $k=K-2$ is illustrated in Fig.~\ref{upper_bound_figure}.
The light shaded region is the inner bound. The outer bound is the union of the light
and the dark shaded regions. 
The boundary of the optimum stopping set falls within
the dark shaded region.
For any $p$ within the inner bound we know that it is optimal to stop
and for any $p$ outside the outer bound it is optimal to continue.
We are uncertain about the optimal action for belief states within the 
dark shaded region.

\section{Optimum Relay Selection in a Simplified Model}
\label{section:simplified_model}
The bounds obtained in the previous section require us to compute
the threshold functions $\{\phi_\ell: \ell=0,1,\cdots,K-1\}$ (see 
Definition~\ref{phi_threshold_definition}) recursively. These are computationally
very intensive to obtain. Hence, in this section we simplify the exact model and
extract a simple selection rule. Our aim is to apply this simple rule to the
exact model and compare its performance with the other policies. 

\subsection{The Simplified Model}
Now we describe our \emph{simplified model}. There are $\tilde{N}$ relays. 
Here, $\tilde{N}$ is a constant and is known to the source. 
The key simplification in this model is that here the relay 
nodes wake-up \emph{at the first $\tilde{N}$ points of a Poisson process
of rate $\frac{\tilde{N}}{T}$}. The following are the motivations for considering such a simplification.
Note that in our actual model (Section~\ref{system_model}), when $N=\tilde{N}$, the inter wake-up times
$\{U_k: 1\le k\le \tilde{N}\}$ are identically distributed \cite[Chapter 2]{orderstatistics}, but not independent.
Their common cdf (cumulative distribution function)  is  $F_{U_k|N}(u|\tilde{N})=1-{(1-\frac{u}{T})}^{\tilde{N}}$ for $u\in(0,T)$.
From Fig.~\ref{cdf_figu} we observe that the cdf of $\{U_k: 1\le k\le \tilde{N}\}$ is close to that of an exponential random variable of
parameter $\frac{\tilde{N}}{T}$ and the approximation becomes better for large values of $\tilde{N}$ (for a fixed $T$).
This motivates us to approximate the actual inter wake-up times by exponential random variable of rate $\frac{\tilde{N}}{T}$.
Further in the simplified model we allow the inter wake-up times to be independent.
Finally, observe that in the simplified model 
 the average number of relays that wake-up within the duty cycle
$T$ is $\tilde{N}$ which is same as that in the exact model when $N=\tilde{N}$. 

\begin{figure}[ht]
\centering
\subfigure[]{
\includegraphics[scale=0.3]{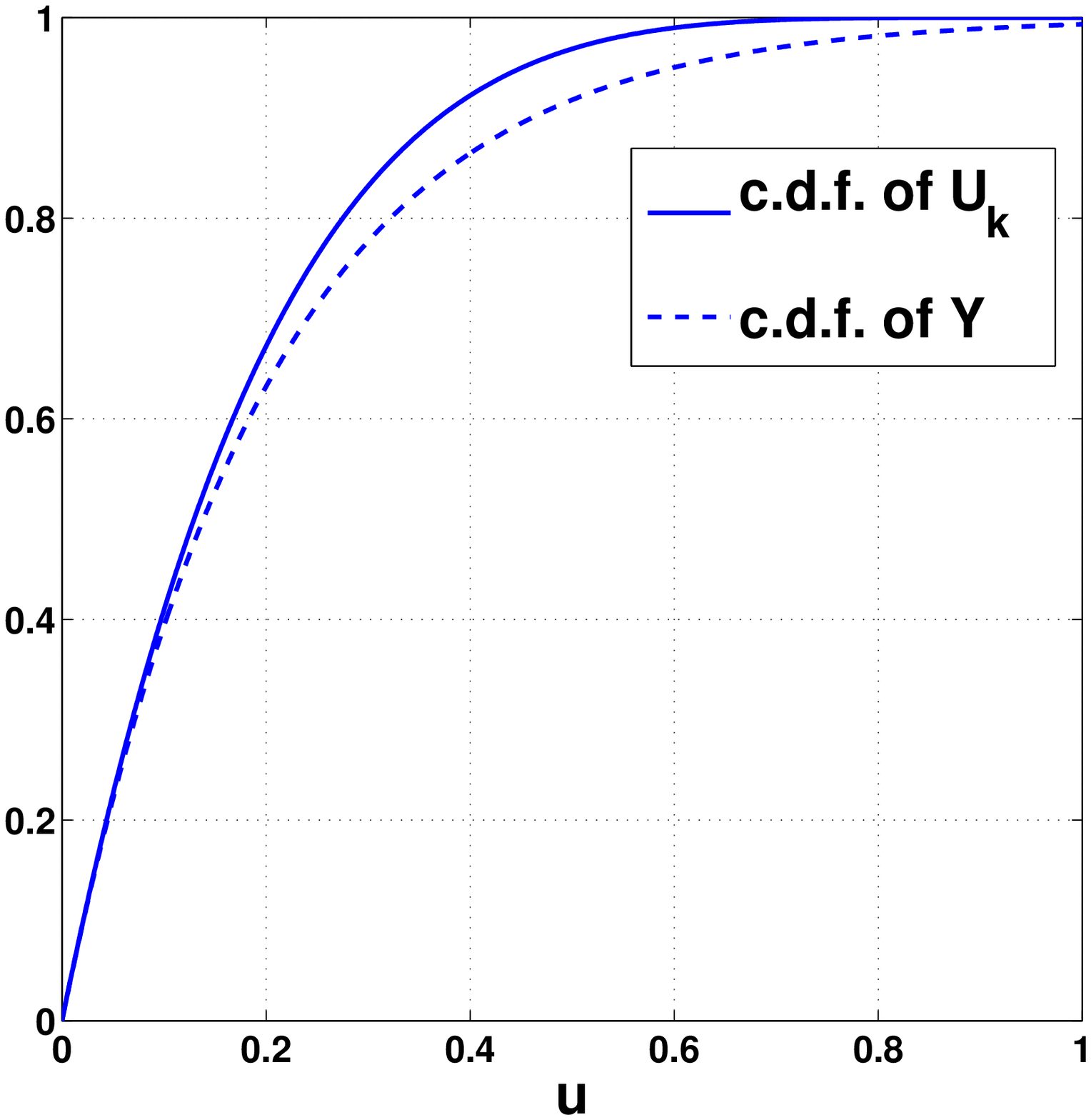}
\label{cdf1}
}
\subfigure[]{
\includegraphics[scale=0.3]{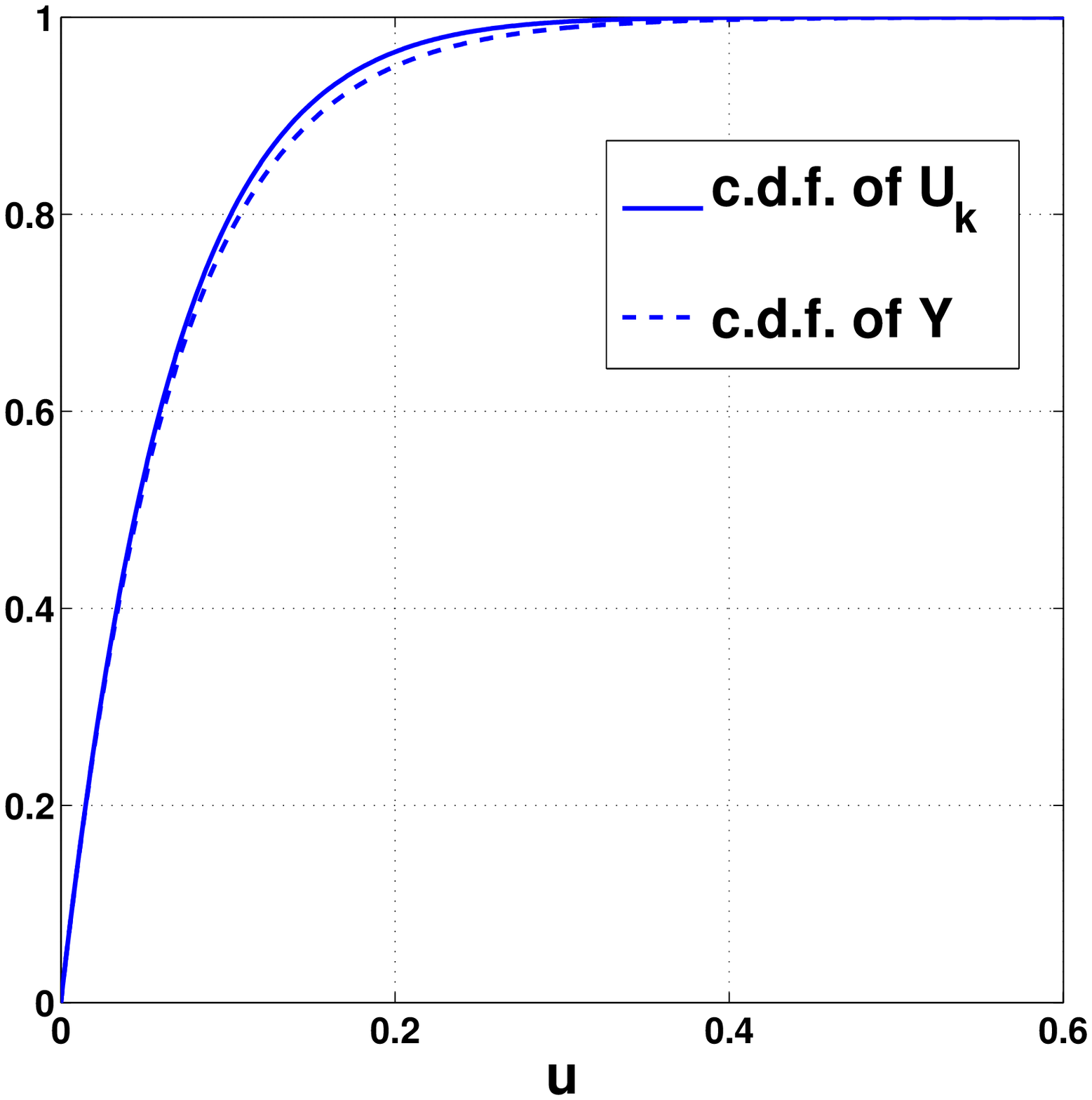}
\label{cdf2}
}
\caption{The cdfs  $F_{U_k|N}(.|\tilde{N})$ and $F_Y(.)$ 
where $Y\sim Exponential(\frac{\tilde{N}}{T})$ (with $T=1$) 
are plotted for \subref{cdf1} $\tilde{N}=5$ and 
\subref{cdf2} $\tilde{N}=15$.\label{cdf_figu}}
\end{figure}

We will use the  notations such as $\tilde{W}_k, \tilde{R}_k, \tilde{U}_k$, etc., to represent the 
analogous quantities that were defined for the exact model. For instance, $\tilde{W}_k$ represents the wake-up
time of the $k$ th relay. However,  unlike in the exact model, here $\tilde{W}_k$
can be beyond $T$.
As mentioned before,
$\{\tilde{U}_k:k=1,\cdots,\tilde{N}\}$ are simply iid exponential random variables with parameter $\frac{\tilde{N}}{T}$.
$\{\tilde{R}_k:k=1,\cdots,\tilde{N}\}$ are iid random rewards with common pdf $f_R$ which is same 
as that in the exact model.

\subsection{MDP Formulation}
Again, here the decision instants are the times at which the relays wake-up. At some stage $k$, $1\le k<\tilde{N}$,
suppose $(\tilde{W}_k,\tilde{R}_k)=(w,b)$ then the one step cost of stopping is $-\eta b$ and that of
continuing is $\tilde{U}_{k+1}$. Note that since $\tilde{U}_{k+1}\sim Exp(\frac{\tilde{N}}{T})$, 
the one step costs do not depend on $w$, 
which means that the optimal policy for the simplified model 
does not depend on the value of $w$. Also since the number of 
relays $\tilde{N}$ is a contant, we do not wish to retain it 
as a part of the state unlike that in the actual state space 
$\mathcal{S}_k^a$ (Equation~(\ref{actual_state_space_equn})).
Therefore we simplify the state space to be 
$\tilde{\mathcal{S}}_0=\{0\}$ and for $k=1,2,\cdots,\tilde{N}$,
\begin{eqnarray*}
 \tilde{\mathcal{S}}_k=[0,\overline{R}]\cup\{\psi\}\mbox{.}
\end{eqnarray*}
As before $\psi$ is the terminating state. Suppose at 
some stage $1\le k<\tilde{N}$ the state is $\tilde{B}_k=b$
then the next state $s_{k+1}$ will be
\begin{eqnarray*}
 s_{k+1}=\left\{\begin{array}{ll}
                 \psi&\mbox{ if } a_k=1\\
		\max\{b,\tilde{R}_{k+1}\}&\mbox{ if } a_k=0
                \end{array}\right.\mbox{.}
\end{eqnarray*}
We had mentioned the one step costs earlier. We  write them down here for the sake of completeness,
\begin{eqnarray*}
 \tilde{g}_k\Big(b,a_k\Big)=\left\{\begin{array}{ll}
                                      -\eta b&\mbox{ if } a_k=1\\
				 	\tilde{U}_{k+1}&\mbox{ if } a_k=0
                                     \end{array}\right.\mbox{.}
\end{eqnarray*}
The cost of termination is simply $\tilde{g}_{\tilde{N}}(b)=-\eta b$.

\subsection{Optimal Policy via One-Step-Stopping Set}
In this section we will prove that the \emph{one-step-look-ahead rule} is 
optimal for the simplified model. The idea is to show that the 
\emph{one-step-stopping set is absorbing} \cite[Section~4.4]{optimalcontrol}. 
All these will now be defined. For an alternate
derivation of the optimal policy by value iteration, 
see the next section (Section~\ref{solution_value_iteration_section}). 

At stage $k$, $1\le k< \tilde{N}$, when the state is $b$, the cost of stopping is 
simply $c_s(b)=-\eta b$. The cost of continuing for one more step (which is $\tilde{U}_{k+1}$) and then stopping
at the next stage (where the state is $\max\{b,\tilde{R}_{k+1}\}$) is,
\begin{eqnarray*}
c_c(b)&=&\mathbb{E}\Big[\tilde{U}_{k+1}-\eta\max\{b,\tilde{R}_{k+1}\}\Big]\\
&=&-\eta\Big(\mathbb{E}[\max\{b,R\}]-\frac{T}{\eta\tilde{N}}\Big)
\end{eqnarray*}
By defining the function $\beta(\cdot)$ for $b\in[0,\overline{R}]$ as
\begin{eqnarray}
\label{beta_1_equn}
\beta(b)&=&\mathbb{E}\Big[\max\{b,R\}\Big]-\frac{T}{\eta\tilde{N}}\mbox{,}
\end{eqnarray}
we can write $c_c(b)=-\eta\beta(b)$.
Note that both the costs, $c_s$ and $c_c$, do not depend on the stage index $k$.
\begin{definition}
We define the \emph{One-step-stopping set} as, 
\begin{eqnarray}
\label{one_step_equn}
\mathcal{C}_{1step}&=&\Big\{b\in[0,\overline{R}]: -\eta b\le -\eta \beta(b)\Big\}.
\end{eqnarray}
i.e., it is the set of all states $b\in[0,\overline{R}]$
where the cost of stopping, $c_s(b)$, is less than the cost of continuing for
one more step and then stopping at the next stage $c_c(b)$.
\hfill $\blacksquare$
\end{definition}

We will show that $\mathcal{C}_{1step}$ is characterized by a
threshold $\alpha$ and can be written as $\mathcal{C}_{1step}=[\alpha,\overline{R}]$.
This will require the following properties about $\beta(\cdot)$.
\begin{lemma}
 \label{beta_1_lemma}
\verb11
\begin{enumerate}
  \item $\beta$ is continuous, increasing and convex in $b$.
  \item If $\beta(0)<0$, then $\beta(b)<b$ for all $b\in[0,\overline{R}]$.
  \item If $\beta(0)\ge 0$, then $\exists$ a unique $\alpha$
    such that $\alpha=\beta(\alpha)$.
  \item If $\beta(0)\ge0$, then $\beta(b)<b$ for
    $b\in(\alpha,\overline{R}]$ and $\beta(b)>b$ for
    $b\in[0,\alpha)$.
  \end{enumerate}
\end{lemma}
\begin{proof}
See Appendix~\ref{beta_1_lemma_proof_appendix}.
\end{proof}

\emph{Discussion of Lemma~\ref{beta_1_lemma}:} When $\beta(0)\ge0$ then using Lemma~\ref{beta_1_lemma}.3
and \ref{beta_1_lemma}.4, we can write $\mathcal{C}_{1step}$ in (\ref{one_step_equn}) 
as $\mathcal{C}_{1step}=[\alpha,\overline{R}]$.
For the other case
where $\beta(0)<0$, from Lemma~\ref{beta_1_lemma}.2
it follows that $\mathcal{C}_{1step}=[0,\overline{R}]$.
Thus by defining $\alpha=0$ whenever $\beta(0)<0$ we can write
$\mathcal{C}_{1step}=[\alpha,\overline{R}]$ for either case.
\hfill $\blacksquare$

\begin{definition}
Depending on the value of $\beta(0)$ define $\alpha$ as follows,
\begin{eqnarray}
\label{equn:beta_definition}
 \alpha=\left\{\begin{array}{ll}
                \beta_1(\alpha)&\mbox{ if } \beta_1(0)\ge0\\
		0&\mbox{ otherwise }
               \end{array}\right.
\end{eqnarray}
\hfill $\blacksquare$
\end{definition}

\begin{definition}
A policy is said to be \emph{one-step-look-ahead} if at stage $k$, $1\le k<\tilde{N}$,
it stops iff the state $b\in\mathcal{C}_{1step}$, i.e., iff the cost of stopping, $c_s(b)$,
is less than the cost of continuing for one more step and then stopping, $c_c(b)$.
\hfill $\blacksquare$
\end{definition}

\begin{definition}
Let $\mathcal{C}$ be some subset of the state space 
$[0,\overline{R}]$, i.e., $\mathcal{C}\subseteq[0,\overline{R}]$.
We say that $\mathcal{C}$ is \emph{absorbing} if for every $b\in\mathcal{C}$, 
if the action at stage $k$, $1\le k<\tilde{N}$, is to continue, then the next state,
$s_{k+1}$ at stage $k+1$, also falls into $\mathcal{C}$.
\hfill $\blacksquare$
\end{definition}

Since we have expressed $\mathcal{C}_{1step}$ as $[\alpha,\overline{R}]$ and since
$s_{k+1}=\max\{b,\tilde{R}_{k+1}\}$ it is clear that $\mathcal{C}_{1step}$
is absorbing. Finally, referring to \cite[Section~4.4]{optimalcontrol}, it follows that,
for optimal stopping problems, \emph{whenever the one-step-stopping set is absorbing then
the one-step-look-ahead rule is optimal}. Thus the optimal policy for the simplified model is to
choose the first relay whose reward is more than $\alpha$.
If none of the relays' reward values are more than $\alpha$ then at the last stage, $\tilde{N}$,
choose the one with the maximum reward. 

\subsection{Optimal Policy via Value Iteration}
\label{solution_value_iteration_section}
In this section we provide an alternative derivation for the optimal policy
(already obtained in the previous section). 
We will write down the value functions starting from the last stage $\tilde{N}$ and
proceed backwards, and then simplify to obtain the optimal policy.

The value function for the last stage $\tilde{N}$ is simply
$\tilde{J}_{\tilde{N}}(b)=\tilde{g}_{\tilde{N}}(b)=-\eta b$.
Next, when the stage is $\tilde{N}-1$,
\begin{eqnarray}
\label{Js_last_1_equn}
 \tilde{J}_{\tilde{N}-1}(b)&=&\min\bigg\{-\eta b, \mathbb{E}\Big[\tilde{U}_{\tilde{N}}+
\tilde{J}_{\tilde{N}}\Big(\max\{b,\tilde{R}_{\tilde{N}}\}\Big)\Big]\bigg\}\nonumber\\
&=&\min\bigg\{-\eta b,\mathbb{E}\Big[\tilde{U}_{\tilde{N}}-\eta \max\{b,\tilde{R}_{\tilde{N}}\Big] \bigg\}\nonumber\\
&=&\min\left\{-\eta b,-\eta \left(\mathbb{E}[\max\{b,R\}]-\frac{T}{\eta\tilde{N}}\right)\right\}\nonumber\\
&=&\min\Big\{-\eta b, -\eta \beta_1(b)\Big\}\mbox{,}
\end{eqnarray}
where the function $\beta_1(\cdot)$ is exactly same as the function $\beta(\cdot)$ in (\ref{beta_1_equn}),
which we reproduce here for convenience,
\begin{eqnarray*}
\beta_1(b)&=&\mathbb{E}\Big[\max\{b,R\}\Big]-\frac{T}{\eta\tilde{N}}\mbox{.}
\end{eqnarray*}
$\beta_1$ satisfies the properties listed in Lemma~\ref{beta_1_lemma}.

From (\ref{Js_last_1_equn}) it is clear that at stage $\tilde{N}-1$ the optimal 
policy is to stop iff $-\eta b\le -\eta \beta_1(b)$, i.e., iff $b\ge\beta_1(b)$.
%
Whenever $\beta_1(0)<0$, from Lemma~\ref{beta_1_lemma}.2 
and (\ref{Js_last_1_equn}), we observe that at stage $\tilde{N}-1$ 
it is optimal to stop for any $b\in[0,\overline{R}]$. On the otherhand 
when $\beta_1(0)\ge0$, from Lemma~\ref{beta_1_lemma}.3, \ref{beta_1_lemma}.4 
and (\ref{Js_last_1_equn}), we can conclude that it is optimal to stop iff $b\ge\alpha$. 
A plot of the function $\beta_1(\cdot)$ for the case when $\beta_1(0)\ge0$ is shown in 
Fig.~\ref{beta_behaviour_figure}.
It will follow  that there
is a similar function at each stage. Formally, at stage $k$ there 
is a function $\beta_{K-k}(\cdot)$ such that at stage $k$ it is optimal to stop
iff $b\ge\beta_{K-k}(b)$.
Further $\beta_{K-k}(\cdot)$ statisfies for $b<\alpha$, $\beta_{K-k}(b)\ge\beta_1(b)$
and for $b\ge\alpha$, $\beta_{K-k}(b)=\beta_1(b)$. This property of the $\beta$ functions
is illustrated in Fig.~\ref{beta_behaviour_figure} for stages $K-2$ and $K-3$.
Thus the  
optimal policy at any other stage 
$k=1,2,\cdots,\tilde{N}-2$, is same as the above mentioned $\alpha$-threshold policy. 

\begin{figure}[h]
\centering
\includegraphics[scale=0.35]{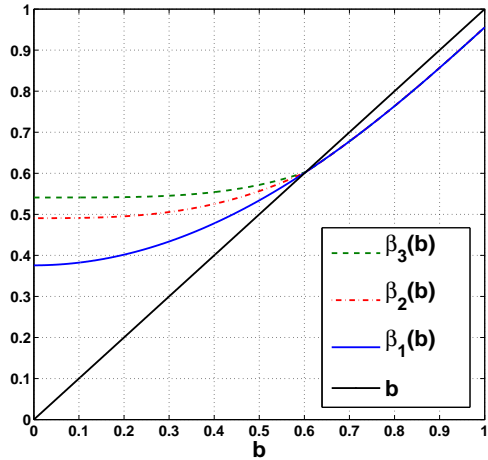} 
\caption{Simplified Model: Illustration of the sucessive $\beta$ functions. 
The threshold $\alpha$ is the point of intersection of $\beta_1(b)$ with the
linear function, $b$. In the figure, $\alpha=0.6$.
$\beta_{\ell}(b)$ as a function of $\ell$ is increasing for $b<\alpha$.
For $b\ge\alpha$, $\beta_{\ell}(b)=\beta_1(b)$.}
\label{beta_behaviour_figure}
\end{figure}

First we will extend the definition of $\alpha$ for the case when 
$\beta_1(0)<0$ by defining $\alpha=0$ (which is same as the definition 
of $\alpha$ in (\ref{equn:beta_definition}) in the previous section).

\begin{definition}
Depending on the value of $\beta_1(0)$ define $\alpha$ as follows,
\begin{eqnarray*}
 \alpha=\left\{\begin{array}{ll}
                \beta_1(\alpha)&\mbox{ if } \beta_1(0)\ge0\\
		0&\mbox{ otherwise }
               \end{array}\right.
\end{eqnarray*}
\hfill $\blacksquare$
\end{definition}

\begin{lemma}
\label{Js_beta_k_lemma}
For every  $k\in\{1,2,\cdots,\tilde{N}-1\}$ the following holds,
\begin{eqnarray}
	\label{Js_k_equn}
	\tilde{J}_{k}(b)&=&\min\Big\{-\eta b,-\eta\beta_{\tilde{N}-k}(b)\Big\}\mbox{,}
\end{eqnarray}
where $\beta_1(b)$ is as defined in (\ref{beta_1_equn}) and for 
$k=1,2,\cdots,\tilde{N}-2$,
\begin{eqnarray}
\label{beta_k_equn}
\beta_{\tilde{N}-k}(b)&=&\mathbb{E}\bigg[\max\Big\{b,R,\beta_{\tilde{N}-(k+1)}(\max\{b,R\})\Big\}\bigg]
- \frac{T}{\eta \tilde{N}}\mbox{,}
\end{eqnarray}
and has the property, $\beta_{\tilde{N}-k}(b)\ge\beta_{\tilde{N}-(k+1)}(b)$ for any $b\in[0,\overline{R}]$. In particular if 
 $b\ge\alpha$ then $\beta_{\tilde{N}-k}(b)=\beta_1(b)$.
\end{lemma}
\begin{proof}
Here we provide only an outline of the proof. For a complete 
proof, see Appendix~\ref{Js_beta_k_lemma_proof_appendix}.
The result already holds for $k=\tilde{N}-1$ (see (\ref{Js_last_1_equn}) and (\ref{beta_1_equn})).
Next we prove the result for $\tilde{N}-2$. The proof is by induction.
Suppose for some $k$, $1<k\le\tilde{N}-2$, (\ref{Js_k_equn}) and (\ref{beta_k_equn}) hold  
along with the ordering property mentioned in the Lemma. We write down the value function $\tilde{J}_{k-1}$ in terms
of $\tilde{J}_{k}$ and straight forward manipulation will yield (\ref{Js_k_equn}) and (\ref{beta_k_equn}) for $k-1$.
The ordering result for $k-1$ can also be easily obtained by using the ordering result for $k$.  
In Fig.~\ref{beta_behaviour_figure} we have depicted this ordering behaviour of the $\beta_{\ell}$ functions.
\end{proof}

\noindent
The following main theorem is a simple consequence 
of the Lemma \ref{beta_1_lemma} and Lemma \ref{Js_beta_k_lemma}.
\begin{theorem}
\label{simple_policy_theorem}
 At any stage $k=1,2,\cdots,\tilde{N}-1$ the optimal policy for the simplified model is to stop iff $\tilde{B}_k=b\ge\alpha$.
\end{theorem}
\begin{proof}
 From (\ref{Js_k_equn}) in Lemma \ref{Js_beta_k_lemma}, it follows that the 
optimal policy is to stop iff $-\eta b\le-\eta\beta_{\tilde{N}-k}(b)$ i.e., $b\ge\beta_{\tilde{N}-k}(b)$.
 If $b\ge\alpha$ then from  Lemma \ref{beta_1_lemma}.4 and Lemma \ref{Js_beta_k_lemma} 
we have $b\ge\beta_1(b)=\beta_{\tilde{N}-k}(b)$ and hence it is optimal to stop (see Fig. \ref{beta_behaviour_figure} for an illustrations).
On the otherhand if $b<\alpha$ then (again from Lemma \ref{beta_1_lemma}.4 and 
\ref{Js_beta_k_lemma}) we have $b<\beta_1(b)\le\beta_{\tilde{N}-k}(b)$ and 
hence the optimal action is to continue.
\end{proof}

Thus the policy for the simplified model is to simply select  the first relay  with  a reward of more that $\alpha$. If all 
the relays have reward of less than $\alpha$ then at the last stage $\tilde{N}$, choose the one
with the best reward.

\subsection{Analysis of the $\alpha$-Threshold Policies}
We have thus seen that the optimal policy for the simplified model is characterized by a 
threshold $\alpha$. Let $R_\alpha$ represent the 
reward obtained when the threshold used is $\alpha$. $R_\alpha$ is equal 
to the reward value of that relay to which the packet is finally forwarded.
We are interested in obtaining an expression for $\mathbb{E}[R_\alpha]$
(this will be useful later in Section~\ref{end_to_end_section}).
$\mathbb{E}[R_\alpha]$ can be written down as
\begin{eqnarray}
\label{R_expectation_equn}
\mathbb{E}[R_\alpha]&=&\int_0^{\overline{R}}\mathbb{P}(R_\alpha>r)dr\mbox{,}
\end{eqnarray}
which will require us to obtain $\mathbb{P}(R_\alpha>r)$ for $r\in[0,\overline{R}]$.
Let us consider two cases, $r\in[0,\alpha]$ and $r\in(\alpha,\overline{R}]$.

For $r\in[0,\alpha]$, the average reward $R_\alpha>r$ whenever there is at least one relay 
with a reward value of more than $r$. Therefore for $r\in[0,\alpha]$,
\begin{eqnarray}
 \label{R_ccdf1_equn}
 \mathbb{P}(R_\alpha>r) &=& \mathbb{P}(\max\{\tilde{R}_1,\cdots,\tilde{R}_{\tilde{N}}\}>r)\nonumber\\
 &=& 1-\mathbb{P}(\max\{\tilde{R}_1,\cdots,\tilde{R}_{\tilde{N}}\}\le r)\nonumber\\
 &=& 1-{F_R(r)}^{\tilde{N}}\mbox{.}
\end{eqnarray}
 The third equality is because the $\tilde{R}_i$'s
are iid with $F_R$ being their common cdf.

Now for $r\in(\alpha,\overline{R}]$, the average reward $R_\alpha>r$ 
whenever the \emph{set of relays whose rewards are more than $\alpha$} is 
nonempty and further the reward of the first relay to wake-up from 
this set is more than $r$. Therefore for $r\in(\alpha,\overline{R}]$,
\begin{eqnarray}
 \label{R_ccdf2_equn}
\mathbb{P}(R_\alpha>r)&=&\Big(1-{F_R(\alpha)}^{\tilde{N}}\Big)\frac{1-F_R(r)}{1-F_R(\alpha)}\mbox{.}
\end{eqnarray}
$1-{F_R(\alpha)}^{\tilde{N}}$ is the probability that there is at least one 
relay with a reward value of more than $\alpha$ and $\frac{1-F_R(r)}{1-F_R(\alpha)}$ is the probability
that the reward of the first relay (to wake-up from the set mentioned above) 
is more than $r$ conditioned on the fact that its  reward is already more than $\alpha$.

Using (\ref{R_ccdf1_equn}) and (\ref{R_ccdf2_equn}) in (\ref{R_expectation_equn})
it is possible to numerically compute $\mathbb{E}[R_\alpha]$. 
We will use these expressions while describing a policy $\pi_{A-SIMPL}$ (in Section~\ref{end_to_end_section})
which is derived from the simplified model.
For $\alpha_1>\alpha_2$ it is clear than $R_{\alpha_1}\ge R_{\alpha_2}$ which means that $\mathbb{E}[R_\alpha]$
as a function of $\alpha$ is non decreasing.

\section{Numerical and Simulation Results}
\label{simulation_results}
\subsection{One Hop Performance}
\label{one_hop_performance}
Recall (from Section~\ref{system_model}) that our model admits any general 
reward associated with a relay. 
In this section we  perform and discuss a simulation study
of geographical forwarding in a dense sensor network with sleep-wake cycling nodes
where the reward provided by a relay is the progress made towards the base-station
(or sink) if the packet is forwarded to that relay.
 In Appendix~\ref{simulation_cases_appendix} 
we have shown simulation results for 
other rewards  (e.g., reward being a function of the progress and
channel gain).

\begin{figure}[h]
 \centering
\includegraphics[scale=0.4]{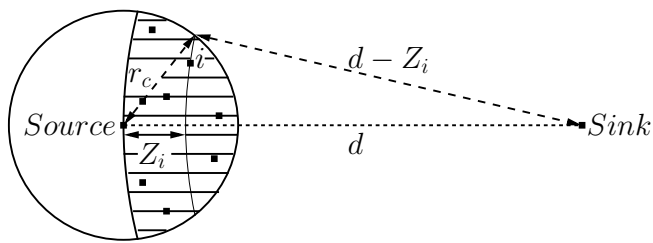}
\caption{\label{forwarding_set_figure} The hatched region 
 is the forwarding region.}
\end{figure}

The source and
sink are separated by a distance of $d=10$ (see Fig.~\ref{forwarding_set_figure}). The source has a packet to
forward at time $0$. The communication radius of the source is
$r_c=1$.  The potential relay nodes are the neighbors of the source
that are closer to the sink than itself.  The period of sleep-wake
cycling is $T=1$.  Let $Z_i$ represent the progress of relay $i$. $Z_i$ is the 
difference between the
source-sink and relay-sink distances.  
The reward associated with a relay $i$ is simply the progress made by it, i.e., $R_i=Z_i$. 
We interchangeably use progress and reward in this section.

Each of the nodes is located uniformly in the forwarding set,
independently of the other nodes. Therefore, it can be shown that, the progress made by them
are iid with pdf
\begin{eqnarray}
\label{equn:progress_pmf}
 f_{Z}(r)&=&\frac{2(d-r)cos^{-1}\left(\frac{d^2+{(d-r)}^2-{r_c}^2}{2d(d-r)}\right)}{\mbox{Area of the forwarding region}}\mbox{,}
\end{eqnarray}
and the support of $f_{Z}$ is $[0,r_c]$. Hence $r_c$ 
is analogous to $\overline{R}$ (see System Model, Section~\ref{system_model}).  We take the bound 
on the number of relays as $K=50$, and
the initial pmf is taken as truncated Poisson with 
parameter $10$, {i.e.,} for $n=1,2,\cdots,K$, $p_0(n)=c\frac{10^n}{n!}e^{-10}$ where $c$ is the 
normalization constant. The above mentioned reward pmf ($f_Z$) and initial belief ($p_0$) will be a good approximation
if the nodes are deployed in a region according to a spatial Poisson process of rate $10$. 
The approximation will become better for larger values of $K$.

Since it is computationally intensive to obtain the thresholds
$\{\phi_l\}$ in (\ref{phi_k_equn}) inductively, we have
discretized the space $[0,T]\times[0,\overline{R}]$ into
$100\times100$ equally spaced points and obtain $\{\phi_l\}$ at these
points.  Appropriate pmfs are obtained from the pdfs. All the analysis
in the previous sections hold for this discrete setting as well.

When the actual state space $\mathcal{S}_k^a$ is discrete, then there
are established algorithms to obtain the optimal policy for POMDP
problems \cite{monahan82survey,lovejoy91algorithmic-methods,smallwood-sondik73optimal-control}.
However it is highly computationally intensive to apply these algorithms
here because of the large state space. For instance with $K=50$, the
cardinality of $\mathcal{S}_1^{a}$ is $50\times100\times100$.  Hence
we compare the performance of our suboptimal POMDP policies with the
COMDP policy (Section~\ref{COMDP_section}) that is optimal when the 
actual number of relays is known
 and hence \emph{serves as a lower bound for the cost} that can be
achieved by the optimal POMDP policy.
\subsubsection{Implemented Policies (one-hop)}
\label{implemented_policies}
We summarize the various policies we have implemented.
\begin{itemize}

\item \underline{$\pi_{COMDP}$:} The source knows the actual value of $N$.
Suppose $N=n$, then the source begins with an initial belief with mass
only on $n$. At any stage, $k=1,2,\cdots,n$, if the delay and best
reward pair is $(w,b)$ then transmit if $b\ge\phi_{n-k}(w,b)$,
continue otherwise. See the remark following
Lemma~\ref{one_point_mass_lemma}. 

\item \underline{$\pi_{INNER}$:} We use the inner bound
$\underline{\mathcal{C}}_k(w,b)$ to obtain a suboptimal policy. At stage $k$ if the belief state is
$(p,w,b)$ $(\in\mathcal{S}_k)$, then transmit iff
$p\in\underline{\mathcal{C}}_k(w,b)$.

\item \underline{$\pi_{OUTER}$:} We use the outer bound
$\overline{\mathcal{C}}_k(w,b)$ to obtain a suboptimal policy. 
At stage $k$ if the belief state is $(p,w,b)$ $(\in\mathcal{S}_k)$,
then transmit iff $p\in\overline{\mathcal{C}}_k(w,b)$.

\item \underline{$\pi_{A-COMDP}$}: (Average-COMDP) The source assumes that  $N$ is equal to
its average value $\overline{N}=\left[\mathbb{E}N\right]$ \footnote{[x]
  represents the smallest integer greater than x.}, and begins
with an initial pmf with mass only on $\overline{N}$. Suppose $N=n$,
which the source does not know, then at some stage
$k=1,2,\cdots,\min\{n,\overline{N}\}$ if the delay and best reward
pair is $(w,b)$ then transmit iff $b\ge\phi_{\overline{N}-k}(w,b)$.
In the case when $\overline{N}>n$, if the
source has not transmitted until stage $n$ and further at stage $n$ if the action is to
continue, then since there are no more relays to go, the source ends up
waiting until time $T$ and then forwards to the node with the best
reward.

\item \underline{$\pi_{A-SIMPL}$}: (Average-Simple) This policy is derived from the simplified
model described in Section~\ref{section:simplified_model}. The source considers the
simplified model assuming that there are $\overline{N}=[\mathbb{E}N ]$ number of
relays. It computes the threshold $\alpha$ accordingly using (\ref{equn:beta_definition}). The policy is
to transmit to the first relay that wakes up and offers a reward
(progress in this case) of more than $\alpha$. If there is no such
relay then the source ends up waiting until time $T$, and then
transmits to the node with the best reward.
\end{itemize}

\subsubsection{Discussion}
We have performed simulations to obtain the average values for the
above policies for several values of $\eta$ ranging from $0.1$ to $1000$.  In
Fig.~\ref{delay_figure}, we plot the average delays of the policies
described above as a function of $\eta$.  The average reward is
plotted in Fig.~\ref{reward_figure}.\vspace{-2mm}

\begin{figure}[h!]
 \centering
\subfigure[]{
\includegraphics[scale=0.32]{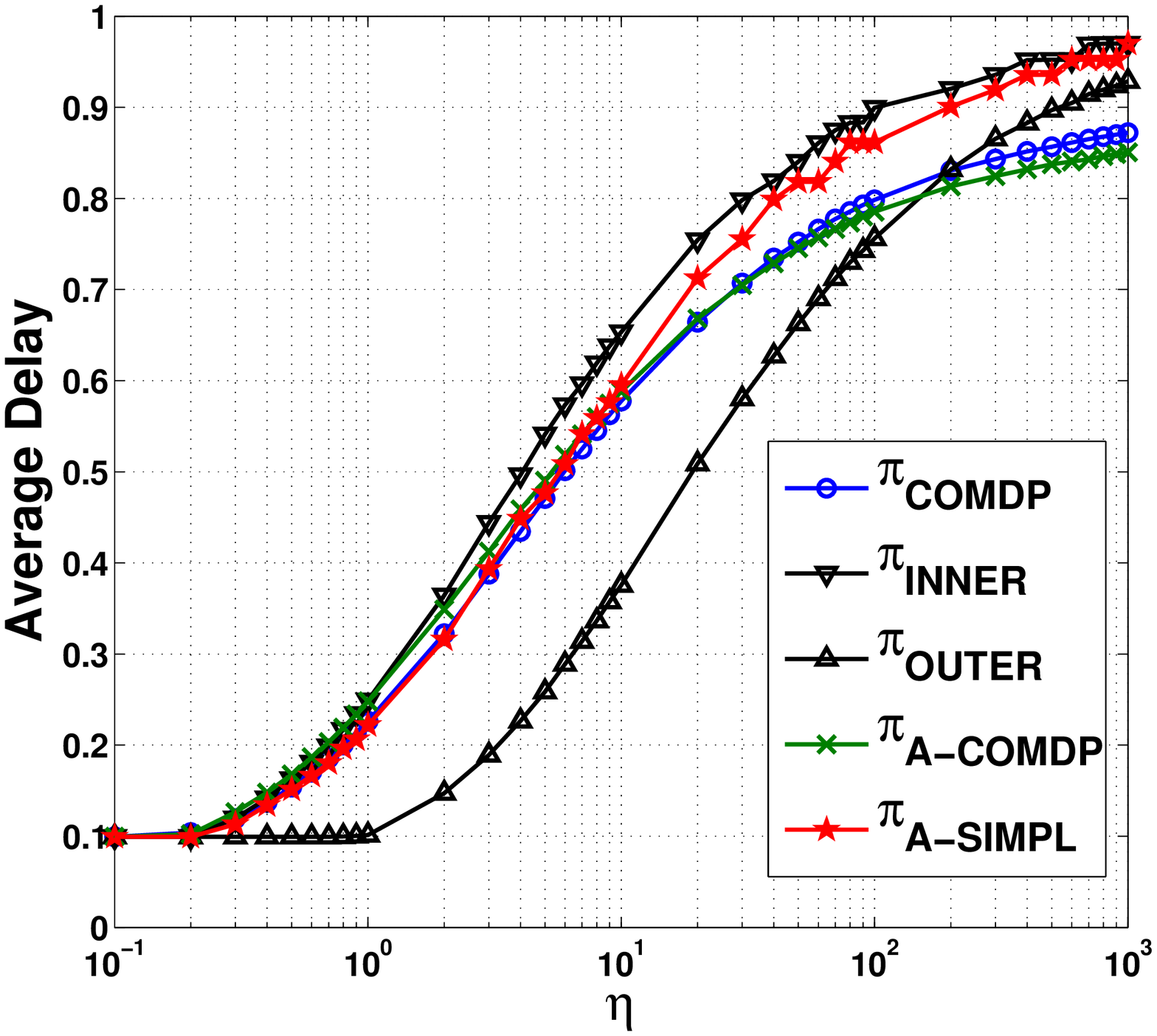}
\label{delay_figure}
}
\centering
\subfigure[]{
\includegraphics[scale=0.32]{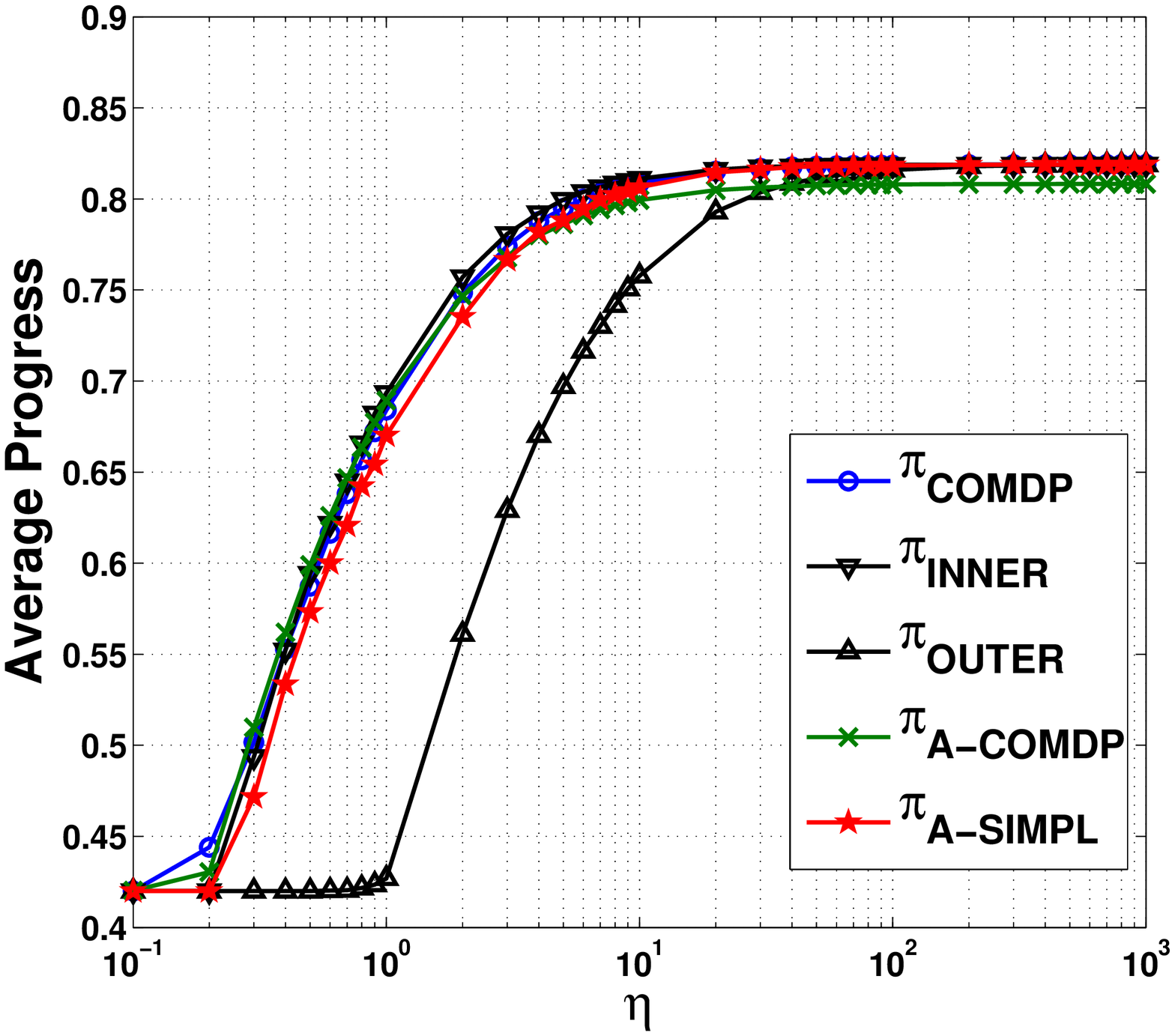}
\label{reward_figure}
}
\caption{\label{delay_reward_figure}
\subref{delay_figure} Average Delay as a function of $\eta$. 
\subref{reward_figure} Average Reward as a function of $\eta$.}
\vspace{-2mm}
\end{figure}

As a function of $\eta$ both the average delay and the average reward are
increasing. This is because for larger $\eta$ we value the progress
more so that we tend to wait for longer time to do better in progress.
For very small values of $\eta$, all the thresholds ($\{\phi_\ell\}$ and
$\alpha$) are very small and most of the time, the packet is forwarded
to the first node (referred to as the \emph{First Forward} policy in \cite{naveen-kumar10geographical-forwarding}).  
For very high values of $\eta$  the policies end up waiting
for all the relays and then choose the one with the best reward 
(referred to as the \emph{Max Forward} policy in \cite{naveen-kumar10geographical-forwarding}). 
Therefore, as $\eta$ increases the average progress of all the policies 
(excluding $\pi_{A-COMDP}$) converge
to $\mathbb{E}[\max\{Z_1,\cdots,Z_N\}]$ which is about $0.82$ 
(see Fig.~\ref{reward_figure}). However
the average progress for $\pi_{A-COMDP}$ converges to a value less than $0.82$.
This is because whenever $\overline{N}<N$ and for large $\eta$ (where all the 
thresholds $\{\phi_\ell\}$ are large)  $\pi_{A-COMDP}$ ends up waiting for the 
first $\overline{N}$ relays and obtain an progress of 
$\max\{Z_1,\cdots,Z_{\overline{N}}\}$ which is less than (or equal to) 
the progress made by the other policies (which is $\max\{Z_1,\cdots,Z_N\}$).
 
Recall that the main problem we are interested in is the one  in
(\ref{main_problem}).  We should be comparing the average delay obtained
using the above policies such that the average reward provided by each
of them is $\gamma$.  This will require us, for each policy, to use an
$\eta$ such that the average reward is equal to $\gamma$. Since we do not have
any closed form expression for average reward in terms of $\eta$, we
proceed as follows. 
 We fix a target $\gamma$.  For each policy, we
choose among the several average reward values (corresponding to the
several $\eta$ values) the one that is closest to the target $\gamma$
and consider the corresponding average delay.  For different target
$\gamma$, in Tables~\ref{table:result_Z} and \ref{table:result_D}  we have tabulated such  
average progress and delay values respectively
for different policies.

\begin{table}[h]
\centering
\begin{tabular}{|c |c |c |c |c|}
\hline
Target $\gamma$&   0.6800  &  0.7200   & 0.7600  &  0.8000\\
\hline
$\mathbb{E}[R_{\pi_{COMDP}}]$ & 0.6840 & 0.7198 & 0.7612 & 0.8000\\
\hline
$\mathbb{E}[R_{\pi_{INNER}}]$ & 0.6822 & 0.7212 & 0.7600 & 0.8001\\
\hline
$\mathbb{E}[R_{\pi_{OUTER}}]$ & 0.6789 & 0.7208 & 0.7578 & 0.8003\\
\hline
$\mathbb{E}[R_{\pi_{A-COMDP}}]$ & 0.6773 & 0.7195 & 0.7590 & 0.8005\\
\hline
$\mathbb{E}[R_{\pi_{A-SIMPL}}]$ & 0.6819 & 0.7165 & 0.7585 & 0.7996\\
\hline
\end{tabular}
\caption{For a given target $\gamma$ (a column) and a policy 
(a row) the entry in the table corresponds to the average progress value
that is closest to the target $\gamma$.}
\label{table:result_Z}
\vspace*{-12mm}
\end{table}

\begin{table}[h!]

\centering
\begin{tabular}{|c |c |c |c |c|}
\hline
Target $\gamma$&   0.6800  &  0.7200   & 0.7600  &  0.8000\\
\hline
$\mathbb{E}[D_{\pi_{COMDP}}]$ & 0.2262 & 0.2711 & 0.3529 & 0.5012\\
\hline
$\mathbb{E}[D_{\pi_{INNER}}]$ & 0.2343 & 0.2905 & 0.3735 & 0.5450\\
\hline
$\mathbb{E}[D_{\pi_{OUTER}}]$ & 0.2359 & 0.2967 & 0.3756 & 0.5551\\
\hline
$\mathbb{E}[D_{\pi_{A-COMDP}}]$ & 0.2336 & 0.2954 & 0.3825 & 0.5997\\
\hline
$\mathbb{E}[D_{\pi_{A-SIMPL}}]$ & 0.2338 & 0.2823 & 0.3684 & 0.5415\\
\hline
\end{tabular}
\caption{For a given target $\gamma$ (a column) and a policy (a row) 
the entry in the table is  the average delay value corresponding to the 
average progress value in Table~\ref{table:result_Z}.}
\label{table:result_D}
\vspace*{-6 mm}
\end{table}

The entries in the first row of both the tables contain different values of 
target $\gamma$ (namely, $0.68$, $0.72$, $0.76$ and $0.8$).
We will discuss the entries in the last column (i.e., entries corresponding to
the target $\gamma$ of $0.8$).
By reading the values from the last column of Table.~\ref{table:result_Z},
which contains the average progress values, we see that the average progress
for all the policies are within $\pm 0.0005$ of $0.8$ (for other columns all the
entries are within $\pm 0.005$ of the corresponding target $\gamma$). 
Hence it is reasonable to compare the delay values of the various policies
in the last column of Table~\ref{table:result_D}. As expected, the COMDP obtains the 
lowest delay (of $0.5012$). There is only a very small performance gap between
the INNER and OUTER bound policies i.e.,  the
delay obtained by the INNER bound policy ($0.5450$) is slightly less than that of
the OUTER bound policy ($0.5551$).
The scheme A-COMDP, which simply assumes that
the actual number of relays is the average of the initial belief, results in
a higher delay (of $0.5997$).
Interestingly we observe that the policy A-SIMPL, which
was derived from the simplified model performs very close to the INNER
bound policy (with an average delay of $0.5415$).
Other columns can be read similarly. 
For small values of target progress, $\gamma$, we see similar performance for all
the policies. 
These observations are for the particular case
where the reward is simply the progress and the initial belief is truncated Poisson.
In Appendix~\ref{simulation_cases_appendix} 
we have shown simulation results for other 
reward structures and initial beliefs. 
We observe similar behavior there as well.  

\subsection{End-to-End Performance}
\label{end_to_end_section}
The single hop problem considered by us was originally motivated by the end-to-end problem.
In the geographical forwarding context, the end-to-end metrics of interest are the total delay and hop count. 
Hop count is important because it is proportional to the number of transmissions and hence the energy expended by the network.
Each of these metrics immediately motivates us to consider two extreme policies. 
One policy is for each node to transmit to its first neighbor in the forwarding set
to wake-up. The second policy is to wait for all the neighbors in 
the forwarding set to wake-up and then transmit to the one
that makes maximum progress towards the sink. It is reasonable to 
expect that the first policy will minimize the end-to-end delay while
the second one will result in the least hop count. Hence there is a tradeoff between the two metrics. 
Suppose we want to minimize the average total end-to-end delay by imposing an average hop count constraint of $h$.
Let $d$ be the distance between the source and the sink. Heuristically, we expect that the hop count 
constraint would be (approximately) met if
each node, enroute to the sink, contributes an average progress of $\frac{d}{h}$. 
For this average progress constraint if each node now uses the locally optimal policy ($\pi_{COMDP}$), 
we expect the average delay at each hop to be minimized and, hence, obtain close to optimal average total delay.
Instead of the optimal policy, each node can use the policy $\pi_{A-SIMPL}$
since its one hop performance is close to the optimum. Also, its application only requires a node $i$
to compute a simple threshold $\alpha_i$, unlike the other policies where
the threshold $\{\phi_\ell\}$ computation is intensive. 
Fig.~\ref{end_to_end_figure} illustrates the multihop forwarding algorithm with each node 
using the locally derived threshold (obtained form the simplified model in Section~\ref{section:simplified_model}) to forward.
Next we briefly describe the network setting and the implemented policies.

\begin{figure}[h!]
 \centering
\includegraphics[scale=0.3]{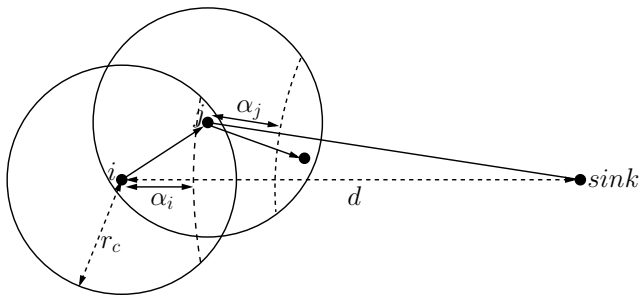}
\caption{Each node enroute to the sink uses the threshold obtained from the simplified model.
$\alpha_i$ and $\alpha_j$ are the thresholds used by nodes 
$i$ and $j$ respectively. \label{end_to_end_figure}}
\end{figure}

\subsubsection{Network Setting}
First we fix a network by placing $M$ nodes randomly in $[0,L]^2$ where $L=10$. $M$ is sampled from
$Poisson(\lambda L^2)$ where $\lambda=5$. Additional source and sink nodes are placed at the locations $(0,0)$ and $(L,L)$ respectively. Further
we have considered a network realization where the forwarding set of each node is nonempty. 
The wake-up times of the nodes are sampled
independently from $Uniform([0,T])$ with $T=1$. If the wake-up instant of a node 
$i$ is $T_i$ then it wakes up at the periodic instances $\{kT+T_i: k\ge0\}$. 
The communication radius of each node is $r_c=1$. 
The source is given a packet at time $0$ and we are interested in routing this packet
to the sink. 

\subsubsection{Implemented Policies (end-to-end)}
We also compare our work with that of Kim et
al.~\cite{kim-etal09optimal-anycast} who have developed end-to-end
delay optimal geographical forwarding in a network setting similar to ours. We
first give a brief description of their work. They minimize, for a
given network, the average delay from any node to the sink when each
node $i$ wakes up asynchronously with rate $r_i$. They show that
periodic wake up patterns obtain minimum delay among all sleep-wake
patterns with the same rate. A relay node with a packet to forward,
transmits a sequence of beacon-ID signals. They propose an algorithm
called LOCAL-OPT \cite{kim-etal08tech-report} which yields, for each
neighbor $j$ of node $i$, an integer $h_j^{(i)}$ such that if $j$
wakes up and listens to the $h$ th beacon signal from node $i$ and if
$h \le h_j^{(i)}$, then $j$ will send an ACK to receive the packet
from $i$. Otherwise (if $h > h_j^{(i)}$) $j$ will go back to sleep and $i$
will continue waiting for further neighbors to wake-up. A
\emph{configuration phase} is required to run the LOCAL-OPT algorithm.

To make a
fair comparision with the work of Kim et al.\ in our network setting  we also introduce beacon-ID
signals of duration $t_I=5$ \emph{msec} and packet transmission duration of
$t_D=30$ \emph{msec}. Description of the policies we have implemented is given below,
\begin{itemize}
 \item $\pi_{FF}$ (First Forward): Each of the node, whenever it gets a packet, it will 
always transmit to the first neighbor in its forwarding set to wake-up, irrespective of the
progress made by it. 

\item $\pi_{MF}$ (Max Forward): We assume that each node knows the number of neighbors in its
forwarding set. 
in this policy a node, when it gets a packet, it will wait for all of its neighbors in the
forwarding set to wake-up. Finally when the last node wakes up, it will forward the packet to the
one which achieves maximum progress towards the sink. 

 \item $\hat{\pi}_{SF}$ (Simplified Forward): This end-to-end policy works by applying the $\pi_{A-SIMPL}$  
policy at each hop. First we fix $\gamma$ as a network parameter 
(as mentioned before, $\gamma$ can be set to $\frac{d}{h}$). 
Nodes \emph{do not know} the number of neighbors in their forwarding set. However they 
know the node density and thus estimates this number as $[\lambda\times\mbox{forwarding set area}]$.
Using this estimated number, a node considers the simplified model and comes up with
a threshold $\alpha$ such that the average progress 
$\mathbb{E}R_\alpha$ in (\ref{R_expectation_equn}) is equal to $\gamma$ 
(see also (\ref{R_ccdf1_equn}) and (\ref{R_ccdf2_equn})). 
$\mathbb{E}R_\alpha$ as a function of 
$\alpha$ is non decreasing. Hence for some node $i$, if $\gamma<\mathbb{E}R_0$ then 
node $i$ chooses its threshold as $0$, and if $\gamma>\mathbb{E}R_{r_c}$ then
node $i$ uses $r_c$ as its threshold.
Suppose node $i$ has a packet to forward.
 When a neighbor of node $i$, say node $j$, wakes up and hears a beacon signal
from $i$, it waits for the ID signal and then sends an ACK signal containing its location information. 
If the progress made by $j$ is more than the threshold, then $i$ forwards
the packet to $j$ (packet duration is $t_D=30$~\emph{msec}). If the progress
made by $j$ is less than the threshold, then $i$ asks $j$ to stay
awake if its progress is the maximum among all the nodes that have
woken up thus far, otherwise $i$ asks $j$ to return to sleep. If more
than one node wakes up during the same beacon signal, then contentions
are resolved by selecting the one which makes the most progress among
them. In the simulation, this happens instantly (as also for the Kim et al. 
algorithm that we compare with); in practice this will require a splitting algorithm; 
see, for example, \cite[Chapter 4.3]{bertsekas-gallager87data-networks}. 
We assume that within $t_I=5$ \emph{msec} all these transactions
(beacon signal, ID, ACK and contention resolution if any) are over.
If there is no eligible
node even after the $\frac{T}{t_I}-th$ beacon signal (one case when
this is possible is when the actual number of nodes $N$ is less than
$[\lambda\times\mbox{Forwarding set Area}]$ and none of the nodes make a
progress of more than the threshold) then $i$ will select one which
makes the maximum progress among all nodes.

\item ${\pi}_{SF}$: This is the same as $\hat{\pi}_{SF}$, but here we assume
that each node knows the exact number of neighbors in its forwarding set and
uses this exact number to come up with the threshold $\alpha$.
Unlike in the previous case, here if none of the neighbors of 
node $i$ make a progress of more that the threshold
used by $i$ then, knowing the number of neighbors, node $i$
choose the neighbor with the best progress when the last one 
wakes up.
$\pi_{FF}$ and $\pi_{MF}$ can be thought of as special cases of
$\pi_{SF}$ with thresholds of $0$ and $r_c$ respectively.

\item \emph{Kim et al.}: We run the LOCAL-OPT algorithm
\cite{kim-etal08tech-report} on the network and obtain the values
$h_j^{(i)}$ for each pair $(i,j)$ where $i$ and $j$ are neighbors. We
use these values to route from source to sink in the presence of sleep
wake cycling. Contentions, if any, are resolved (instantly, in the simulation) by selecting a node
$j$ with the highest $h_j^{(i)}$ index.
\end{itemize}

\begin{figure}
	\centering 
	\includegraphics[scale=0.4] {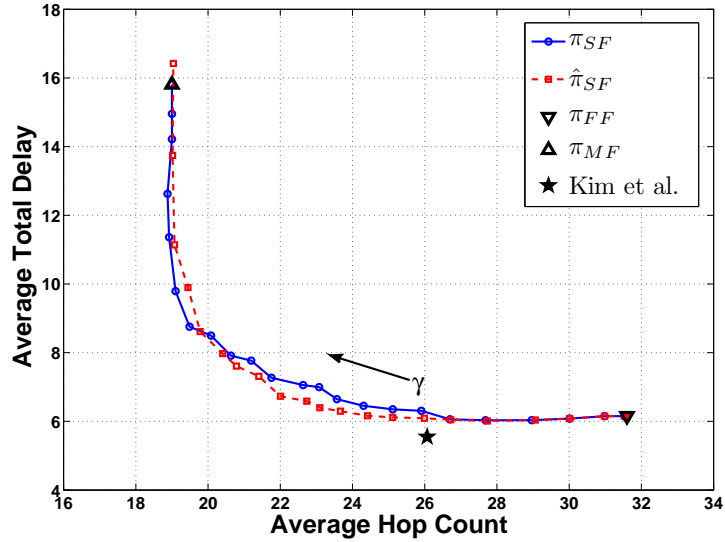}
	\caption{End-to-end performance: Plot of average end-to-end delay vs. average end-to-end hop count
          obtained by applying the simple $\alpha$-threshold policy, $\pi_{A-SIMPL}$, at each hop.
          The  operating points of the policies $\pi_{FF}$, $\pi_{MF}$ and Kim et al.
          are also shown in the figure. Each point on the curve
          corresponds to a different value of $\gamma$ which increases
          along the direction shown. \label{endperformance_figu}}
\vspace*{-6mm}
\end{figure}

\subsubsection{Discussion}
In Fig.~\ref{endperformance_figu} we plot average total delay vs.
average hop count for different policies for fixed node placement, while the averaging is over the wake-up
times of the nodes. Each point on the curve is obtained by averaging  
over 1000 transfers of the packet from the source node to the sink.
 As expected, Kim et al.  achieves minimum average delay.
In comparision with $\pi_{FF}$, Kim et al. also achieve smaller
average hop count. Notice, however that using $\hat{\pi}_{SF}$ (or $\pi_{SF}$) policy and properly choosing
$\gamma$, it is possible to obtain hop count similar to that of Kim et
al., incurring only slightly higher delay. 

The advantage of $\hat{\pi}_{SF}$
over Kim et al. is that there is \emph{no need for a configuration phase}.
Each relay node has to only compute a threshold that depends on the
parameter $\gamma$ which can be set as a network parameter during
deployment. A more interesting approach would be to allow the source
node to set $\gamma$ depending on the type of application. For delay
sensitive applications it is appropriate to use a smaller value of
$\gamma$ so that the delay is small, whereas, for energy constrained
applications (where the network energy needs to conserved) it is
better to use large $\gamma$ so that the number of hops (and
hence the number of transmissions) is reduced. For other applications, moderate
values of $\gamma$ can be used. $\gamma$ can be a part of the ID signal so
that it is made available to the next hop relay.

Another interesting observation from Fig.~\ref{endperformance_figu} is
that the performance of $\hat{\pi}_{SF}$ is close to that of
$\pi_{SF}$. In practice, it may not be possible for a node to
know the exact number of relays in its forwarding set, due to varying channel condition,
node failures, etc. Recall that
$\hat{\pi}_{SF}$ works with the average number of nodes instead of the
actual number. For small values of $\gamma$ both the policies $\pi_{SF}$
and $\hat{\pi}_{SF}$, most of the time, transmit to the first node to
wake up.  Hence the performance is similar for small $\gamma$.  For
large $\gamma$, we observe that the delay incurred by $\hat{\pi}_{SF}$ is larger.

\section{Conclusion}
\label{conclusion}
Our work in this paper was motivated by the problem of geographical forwarding of
packets in a wireless sensor networks whose function is to detect certain infrequent
events and forward these alarms to a base station, and whose nodes are sleep-wake cycling to 
conserve energy. This end-to-end problem gave rise to the local problem faced by
a packet forwarding node, i.e., that  
of choosing one among a set of
potential relays, so as to minimize
the average delay in selecting a relay subject to a
constraint on the average progress (or some reward, in general).
The source does not know the number of available relays, which made
this a sequential decision problem with partial information.
We formulated the problem as a finite horizon POMDP with the unknown
state being the number of available relays. The optimum stopping set
is the set of all pmfs on the number of relays for which the
average cost of stopping is less than that of continuing.  We showed
that the optimum stopping set is convex (Corollary~\ref{convex_corollary}) and obtained threshold points
along certain edges of the simplex which belong to the optimal
stopping set.  A convex combination of these point gave us an inner
bound for the optimum stopping set (Theorem~\ref{inner_bound_theorem}).  We proved a monotonicity result and
obtained an outer bound (Theorem~\ref{outer_bound_theorem}). We also obtained a simple threshold rule by
 formulating an alternate simplified model (Section~\ref{section:simplified_model}).

We have performed simulations to compare the performance of the various policies. 
We observe that the inner bound policy ($\pi_{INNER}$) 
is better than the outer bound ($\pi_{OUTER}$). 
Further the performance of the simple threshold policy ($\pi_{A-SIMPL}$) is comparable with $\pi_{INNER}$,
both of which are close to the optimal policy ($\pi_{COMDP}$).
We have performed one-hop simulations for few other examples where we have considered different
rewards and initial beliefs 
(see Appendix~\ref{simulation_cases_appendix}).
In all the examples, we observe the good performance of the policy $\pi_{A-SIMPL}$.

We have devised simple end-to-end policies ($\pi_{SF}$ and $\hat{\pi}_{SF}$) using $\pi_{A-SIMPL}$.
We have shown that by varying a network parameter these policies can favourably tradeoff between the
average total delay and average hop count. 


\bibliographystyle{IEEEtran}
\bibliography{IEEEabrv,naveen-kumar10relay-selection_tech-report}

\appendices
\section{Proofs of Lemmas in Section~\ref{COMDP_section}}
\label{COMDP_section_proofs_appendix}

\subsection{Proof of Lemma~\ref{phi_lemma_equn}}
\label{phi_lemma_equn_proof_appendix}
From (\ref{phi_k_equn}) (the subscripts of both $U$ and $R$ in the following expressions 
is $K-\ell+1$, which we have suppressed for simplicity),
  \begin{eqnarray*}
   {\phi_\ell(w,b)}
&=&\mathbb{E}_{{K-\ell}}\bigg[\max\bigg\{b,R,\phi_{\ell-1}\Big(w+U,
\max\{b,R\}\Big)\bigg\}-\frac{U}{\eta}\bigg|w,K\bigg]\nonumber\\
&\ge&\mathbb{E}_{K-\ell}\bigg[\max\bigg\{b,R,\phi_{\ell-1}\Big(w+U,
\max\{b,R\}\Big)\bigg\}\bigg|w,K\bigg]-\frac{T-w}{\eta}\nonumber\\
&\ge& b-\frac{T-w}{\eta}\mbox{,}
  \end{eqnarray*}
where the first inequality follows from $U\le T-w$ and the second due to the $\max$ inside the 
expectation. 
\hfill $\blacksquare$

\subsection{Proof of Lemma~\ref{one_point_mass_lemma}}
\label{one_point_mass_lemma_proof_appendix}
 We proceed by value iteration. First we will show that
the lemma holds for $k=n$, where the fixed $n$ could be either less than $K$ or 
equal to $K$ (recall that $K$ is the bound on the number of relays). Suppose $n<K$. 
Since $p_n(n)=1$,  from (\ref{continue_cost_equn}), 
it follows that $c_{n}(p_n,w,b)=T-w-\eta b$. 
Therefore,
\begin{eqnarray*}
 J_{n}(p_n,w,b)=\min\Big\{-\eta b, T-w-\eta b\Big\}
=-\eta b\mbox{.}
\end{eqnarray*}
If $n=K$ then $J_n(p_n,w,b)=g_{K}(p_n,w,b)=-\eta b$. 
Thus for any fixed $n$ we can write
\begin{eqnarray*}
J_{n}(p_n,w,b)= \min\Big\{-\eta b,-\eta \phi_{n-n}(w,b)\Big\}\mbox{.}
\end{eqnarray*}
 Suppose for some $k=1,\cdots,n-1$ the following holds, 
\begin{eqnarray*}
 J_{k+1}(p_{k+1},w,b)&=&\min\Big\{-\eta b,-\eta \phi_{n-(k+1)}(w,b)\Big\}\mbox{.}
\end{eqnarray*}
Then,
\begin{eqnarray*}
c_{k}(p_k,w,b)
&=&\mathbb{E}_{k}\bigg[U_{k+1}+J_{k+1}\Big(\tau_{k+1}(p_k,w,U_{k+1}),w+U_{k+1},
\max\{b,R_{k+1}\}\Big)\bigg|w,n\bigg]\nonumber\\
&=& \mathbb{E}_{k}\bigg[U_{k+1}+\min\bigg\{-\eta \max\{b,R_{k+1}\}, 
-\eta \phi_{n-(k+1)}\Big(w+U_{k+1},\max\{b,R_{k+1}\}\Big)\bigg\}\bigg|w,n\bigg]\nonumber\\
&=& -\eta \mathbb{E}_{k}\bigg[\max\bigg\{b,R_{k+1},\phi_{n-(k+1)}\Big(w+U_{k+1},
\max\{b,R_{k+1}\}\Big)\bigg\}-\frac{U_{k+1}}{\eta}\bigg|w,n\bigg]\mbox{.}
\end{eqnarray*}
In the second equality we have used the induction
hypothesis and the  fact that if $p_k(n)=1$ then \\
$\tau_{k+1}(p_k,w,U_{k+1})(n)=1$. The expectation in
(\ref{c_phi_equn}) is over the pdf $f_R()f_{k}(|w,n)$.  From
(\ref{condorderstat_equation}), note that the pdf $f_{k}(|w,n)$
depends on $k$ and $n$ only through the difference $n-k$. Therefore
$f_{k}(.|w,n)=f_{K-(n-k)}(.|w,K)$.  Using this and (\ref{phi_k_equn})
in (\ref{c_phi_equn}) we can write
\begin{eqnarray}
\label{c_phi_equn}
c_{k}(p_k,w,b)
&=& -\eta \mathbb{E}_{K-(n-k)}\bigg[\max\bigg\{b,R,\phi_{n-(k+1)}\Big(w+U,
\max\{b,R\}\Big)\bigg\}-\frac{U}{\eta}\bigg|w,K\bigg]\nonumber\\
&=&-\eta\phi_{n-k}(w,b)\mbox{.}
\end{eqnarray}
Finally using (\ref{optimal_cost_equn}), we can write, $J_k(p_k,w,b)=\min\Big\{-\eta b,-\eta \phi_{n-k}(w,b)\Big\}$. 
Hence we have proved that the lemma holds for $k$ if it is true for $k+1$. Since we have already shown that the lemma holds for $n$,
from induction argument we can conclude that it holds for all $k=1,2,\cdots,n$.
\hfill $\blacksquare$

\section{Proof of Lemmas in Section~\ref{bounds}}
\label{bounds_section_proofs_appendix}
\subsection{Proof of Lemma~\ref{concave_lemma}}
\label{concave_proof_appendix}
The essence of the proof is same as that in 
\cite[Lemma 1]{porta-etal06point-based}. We 
provide the proof here for completeness. The cost to continue
at stage $K-1$ is (see (\ref{continue_cost_equn})), 
\begin{eqnarray}
\label{c_kbar_equn}
 c_{K-1}(p,w,b)&=&p(K-1)\Big(T-w-\eta b\Big)+p(K)\mathbb{E}_{K-1}\Big[U_{K}-\eta \max\{b,R_{K}\}\Big|w,K\Big]\nonumber\\
&=&p(K-1)\Big(T-w-\eta b\Big)+p(K)\Big(-\eta \phi_1(w,b)\Big)\mbox{.}
\end{eqnarray}
Thus we have shown that $c_{K-1}(\cdot,w,b)$ is an affine function of $p\in\mathcal{P}_{K-1}$, 
for every $(w,b)$. Recalling (\ref{optimal_cost_equn}), $J_{K-1}(\cdot,w,b)$, being the minimum of two affine functions, $-\eta b$ 
and $c_{K-1}(\cdot,w,b)$, is concave on $\mathcal{P}_{K-1}$. The proof now proceeds by induction. 

\emph{Induction hypothesis}: 
For some $k=1,2,\cdots,K-2$, and for each $(w,b)$, $J_{k+1}(\cdot,w,b)$ is concave on $\mathcal{P}_{k+1}$ and can be written down as,
\begin{eqnarray}
\label{equn:s_3_help0}
 J_{k+1}(p,w,b)
&=&\inf_{\alpha\in\mathcal{A}_{k+1}(w,b)}\left<\alpha,p\right>\nonumber\\
&=&\left<\alpha_{k+1}^{(p,w,b)},p\right>\mbox{,}
\end{eqnarray}
where $\mathcal{A}_{k+1}(w,b)$ is some collection of $K-k$ length vectors and 
$\alpha_{k+1}^{(p,w,b)}=\argmin_{\alpha\in\mathcal{A}_{k+1}(w,b)}\left<\alpha,p\right>$.

There are two points to note here. First, in general a concave function can be 
written down as an infimum over some collection of affine functions of the form $\left<\alpha,p\right>+c$
where $c$ is some constant. However, we claim that there are no such constants associated with the $\alpha$
vectors in the set $\mathcal{A}_{k+1}(w,b)$. Second, we are claiming the existence of 
the vector $\alpha_{k+1}^{(p,w,b)}$. Notice that both of these claims are true for stage $K-1$, since 
the set $\mathcal{A}_{K-1}(w,b)$
comprises only two vectors, $\Big((T-w-\eta b),-\eta \phi_1(w,b)\Big)$ and $(-\eta b,-\eta b)$, 
i.e., the induction hypothesis holds for $k=K-1$.

To show that $J_{k}(\cdot,w,b)$ is concave on $\mathcal{P}_k$, it suffices to  prove that 
$c_{k}(\cdot,w,b)$ is concave. $c_k$ in (\ref{continue_cost_equn}) can be written down as,
\begin{eqnarray}
\label{continue_cost_split_equn}
 c_{k}(p,w,b)&=&p(k)\Big(T-w-\eta b\Big)+\sum_{n=k+1}^{K}p(n)\mathbb{E}_{k}\Big[U_{k+1}\Big|w,n\Big]+\nonumber\\
&&\sum_{n=k+1}^{K}
p(n)\mathbb{E}_{k}\bigg[J_{k+1}\Big(\tau_{k+1}(p,w,U_{k+1}),w+U_{k+1},\max\{b,R_{k+1}\}\Big)\bigg|w,n\bigg]\mbox{.}
\end{eqnarray}
Let us focus on the third term in the above summation. Call it $s_3$ for convenience. 
\begin{eqnarray}
\label{equn:s_3_help1}
 s_3&=&\sum_{n=k+1}^{K}p(n)\int_{0}^{\overline{R}}\int_{0}^{T-w}f_R(r)f_{k}(u|w,n)J_{k+1}\Big(\tau_{k+1}(p,w,u),w+u,\max\{b,r\}\Big)dudr\nonumber\\
 &=&\sum_{n=k+1}^{K}p(n)\int_{0}^{\overline{R}}\int_{0}^{T-w}f_R(r)f_{k}(u|w,n)
\left<\alpha_{k+1}^{(\tau_{k+1}(p,w,u),w+u,\max\{b,r\})},\tau_{k+1}(p,w,u)\right>dudr\mbox{.}
\end{eqnarray}
Substituting for $\tau_{k+1}(p,w,u)$ from (\ref{belief_transition_equn}) and  simplifying yields,
\begin{eqnarray}
\label{equn:s_3_help2}
 s_3&=&\sum_{n=k+1}^{K}p(n)\int_{0}^{\overline{R}}\int_{0}^{T-w}f_R(r)f_{k}(u|w,n)
 \sum_{n'=k+1}^{K}\Big(\alpha_{k+1}^{(\tau_{k+1}(p,w,u),w+u,\max\{b,r\})}(n')\Big)
\frac{p(n')f_{k}(u|w,n')}{\sum_{\ell=k+1}^{K}p(\ell)f_k(u|w,\ell)}dudr\nonumber\\
&=&\sum_{n'=k+1}^{K}p(n')\int_{0}^{\overline{R}}\int_{0}^{T-w}f_R(r)f_{k}(u|w,n')
\alpha_{k+1}^{(\tau_{k+1}(p,w,u),w+u,\max\{b,r\})}(n')dudr\mbox{.}
\end{eqnarray}
Define $K-k+1$ length vector $\alpha_k^{(p,w,b)}$ as $\alpha_k^{(p,w,b)}(k)=(T-w-\eta b)$ and for $n=k+1,\cdots,K$,
\begin{eqnarray}
 \alpha_k^{(p,w,b)}(n)=\mathbb{E}_{k}\Big[U_{k+1}\Big|w,n\Big]+\int_{0}^{\overline{R}}\int_{0}^{T-w}f_R(r)f_{k}(u|w,n)
\alpha_{k+1}^{(\tau_{k+1}(p,w,u),w+u,\max\{b,r\})}(n) dudr\mbox{.}
\end{eqnarray}
Then (\ref{continue_cost_split_equn}) can be written as,
\begin{eqnarray*}
 c_k(p,w,b)&=&\left<\alpha_k^{(p,w,b)},p\right>\mbox{.}
\end{eqnarray*}
Now for any $q\ne p$ if we write down $\left<\alpha_k^{(q,w,b)},p\right>$, then it will have a term similar to $s_3$ (see (\ref{equn:s_3_help1})
and (\ref{equn:s_3_help2})), but with 
$\alpha_{k+1}^{(\tau_{k+1}(p,w,u),w+u,\max\{b,r\})}$ replaced with $\alpha_{k+1}^{(\tau_{k+1}(q,w,u),w+u,\max\{b,r\})}$. Let us call
this term as $\hat{s}_3$. 
More precisely $\left<\alpha_k^{(q,w,b)},p\right>$ will be similar to RHS of (\ref{continue_cost_split_equn}), but with the third term there (recall that we had named the 
third term as $s_3$) replaced by $\hat{s}_3$. 
Using (\ref{equn:s_3_help0}) in (\ref{equn:s_3_help1}) we  observe that $\hat{s}_3\ge s_3$ so that,
\begin{eqnarray}
 c_k(p,w,b)\le \left<\alpha_k^{(q,w,b)},p\right>\mbox{.}
\end{eqnarray}
Hence by defining $\mathcal{A}_{k}(w,b):=\{\alpha_k^{(q,w,b)}:q\in\mathcal{P}_k\}$ we can write,
\begin{eqnarray*}
 c_k(p,w,b)
&=&\inf_{\alpha\in\mathcal{A}_k(w,b)}\left<\alpha,p\right>
\end{eqnarray*}
which proves that $c_k(\cdot,w,b)$ is concave. Finally, by including in the set $\mathcal{A}_k(w,b)$,  the $K-k+1$ length
vector with each component equal to $-\eta b$, 
we can express $J_k(p,w,b)$ as,
\begin{eqnarray*}
 J_{k}(p,w,b)
&=&\inf_{\alpha\in\mathcal{A}_{k}(w,b)}\left<\alpha,p\right>\mbox{.}
\end{eqnarray*}
\hfill $\blacksquare$

\subsection{Proof of Lemma~\ref{cost_order_lemma}}
\label{cost_order_proof_appendix}
 Since $q$ has mass only on $k$ and $k+1$, using Lemma~\ref{two_point_mass_lemma} we can write,
\begin{eqnarray*}
 c_k(q,w,b)&=&p(k)\Big(T-w-\eta b\Big)+p(k+1)\Big(-\eta\phi_1(w,b)\Big)\mbox{.}
\end{eqnarray*}
Using  (\ref{condexpectation_equn}) and (\ref{phi_k_equn}), we obtain 
$\phi_1(w,b)=\mathbb{E}\Big[\max\{b,R\}\Big]-\frac{T-w}{2\eta}$.
Substituting for $\phi_1(w,b)$ in the above expression we have,

\begin{eqnarray}
 c_k(q,w,b)&=&p(k)\Big(T-w-\eta b\Big)+p(k+1)\bigg(\frac{T-w}{2}-\eta\mathbb{E}[\max\{b,R\}]\bigg)\mbox{.}
\end{eqnarray}

Recall (\ref{continue_cost_equn}),
\begin{eqnarray*}
 \lefteqn{c_k(p,w,b)=p(k)\Big(T-w-\eta b\Big)+}\\
&&\sum_{n=k+1}^{K}p(n)\mathbb{E}_{k}\bigg[U_{k+1}+J_{k+1}\Big(\tau_{k+1}(p,w,U_{k+1}),
w+U_{k+1},\max\{b,R_{k+1}\}\Big)\bigg|w,n\bigg]\mbox{.}
\end{eqnarray*}
Using (\ref{optimal_cost_equn}) and (\ref{condexpectation_equn}) we can write,
\begin{eqnarray*}
 c_k(p,w,b)&\le&p(k)\Big(T-w-\eta b\Big)+\sum_{n=k+1}^{K}p(n)\mathbb{E}_{k}\Big[U_{k+1}-\eta\max\{b,R_{k+1}\}\Big|w,n\Big]\\
&=&p(k)\Big(T-w-\eta b\Big)+\sum_{n=k+1}^{K}p(n)\bigg(\frac{T-w}{n-k+1}-\eta\mathbb{E}\Big[\max\{b,R\}\Big]\bigg)\\
&\le&p(k)\Big(T-w-\eta b\Big)+\sum_{n=k+1}^{K}p(n)\bigg(\frac{T-w}{2}-\eta\mathbb{E}\Big[\max\{b,R\}\Big]\bigg)\\
&=&p(k)\Big(T-w-\eta b\Big)+\Big(1-p(k)\Big)\bigg(\frac{T-w}{2}-\eta\mathbb{E}\Big[\max\{b,R\}\Big]\bigg)\\
&=&c_k(q,w,b)\mbox{.}
\end{eqnarray*}
\hfill $\blacksquare$


\section{Proof of Lemmas in Section~\ref{section:simplified_model}}
\label{simplified_model_proofs_appendix}

\subsection{Proof of Lemma~\ref{beta_1_lemma}}
\label{beta_1_lemma_proof_appendix}

\noindent
\emph{Proof of \ref{beta_1_lemma}.1:}
 Let $F_R$ represent the cummulative distribution function (cdf) 
of $R$. For  $b\in[0,\overline{R}]$, the cdf of $\max\{b,R\}$ is,
 \begin{equation*}
  F_{\max\{b,R\}}(r)=\left\{\begin{array}{ll}
                     0&\mbox{ if }r<b\\
 		    F_R(r)&\mbox{ if }r\ge b\mbox{,}\end{array}\right.
 \end{equation*}
 using which $\beta(b)$ in (\ref{beta_1_equn}) can be written down as,
 \begin{eqnarray}
  \beta(b)&=&\int_0^{\overline{R}}\Big(1-F_{\max\{b,R\}}(r)\Big)dr-\frac{T}{\eta \tilde{N}}\nonumber\\
 &=&b+\int_b^{\overline{R}}\Big(1-F_R(r)\Big)dr-\frac{T}{\eta \tilde{N}}\mbox{.}\nonumber
 \end{eqnarray}
 ${\beta}'(b)=F_R(b)\ge0$  and ${\beta}''(b)=f_R(b)\ge0$ implies that $\beta$ is continuous, increasing and convex in $b$.\\
 
 \noindent
 \emph{Proof of \ref{beta_1_lemma}.2:}
 From (\ref{beta_1_lemma}) note that  $\beta(\overline{R})<\overline{R}$. 
 Also $\beta$ is convex (from {Lemma} \ref{beta_1_lemma}.1). Hence we can write,
 \begin{eqnarray*}
  \beta(b)&\le&\frac{\overline{R}-b}{\overline{R}}\beta(0)+\frac{b}{\overline{R}}\beta(\overline{R})\\
 &<&b\mbox{.}
 \end{eqnarray*}
 \verb11 \\
 \noindent
 \emph{Proof of \ref{beta_1_lemma}.3:}
   Let $g(b)=b-\beta_1(b)$. Then, $g(0)\le0$ and
   $g(\overline{R})>0$ (because $\beta(\overline{R})<\overline{R}$). Also $g(b)$ is continuous (being differentiable) on
   $[0,\overline{R}]$. Hence, $\exists$ an $\alpha\in[0,\overline{R})$ such that
   $g(\alpha)=0$.

 Suppose $\exists$ an ${\alpha}'>{\alpha}$ such that
 $g({\alpha}')=0$. Then by convexity of $\beta$
 (from {Lemma}~\ref{beta_1_lemma}.1),
 \begin{eqnarray*}
 	\beta({\alpha}')&\le&\frac{\overline{R}-{\alpha}'}{\overline{R}-{\alpha}}\beta(\alpha)+\frac{{\alpha}'-
	{\alpha}}{\overline{R}-{\alpha}}\beta(\overline{R})\mbox{,}
 \end{eqnarray*}
 which implies that ${\beta}(\overline{R})\ge\overline{R}$. Contradicts the fact that, $\beta(\overline{R})<\overline{R}$.\\
 
 \noindent
 \emph{Proof of \ref{beta_1_lemma}.4:}
   Again consider $g(b)=b-\beta(b)$. $g(b)$ is continuous (being
   differentiable) on $[0,\overline{R}]$.  Suppose $\exists$ $b\in(\alpha,\overline{R}]$
   such that $\beta_1(b)>b$, then $g(b)\le0$ and $g(\overline{R})>0$. This implies that
   $\exists$ $b'$ in $[b,\overline{R})$ such that $g(b')=0$. Contradicts the
   uniqueness of $\alpha$ shown in {Lemma} \ref{beta_1_lemma}.3.  Similarly
   it can be shown that $\beta(b)>b$ for $b\in[0,\alpha)$.
\hfill $\blacksquare$ 

\subsection{Proof of Lemma~\ref{Js_beta_k_lemma}}
\label{Js_beta_k_lemma_proof_appendix}
 The proof is by induction. From (\ref{Js_last_1_equn}) and (\ref{beta_1_equn}), we see 
that the result is already true for $k=\tilde{N}-1$.
Next we will prove it for $k=\tilde{N}-2$.
 Let us evaluate the value function at stage $\tilde{N}-2$ and simplify using 
the expression for $\tilde{J}_{\tilde{N}-1}$ (from (\ref{Js_last_1_equn})),
\begin{eqnarray}
	\label{Js_last_2_equn}
	\tilde{J}_{\tilde{N}-2}(b)
	&=&\min\bigg\{-\eta b,\mathbb{E}\Big[\tilde{U}_{\tilde{N}-1}+\tilde{J}_{\tilde{N}-1}(\max\{b,\tilde{R}_{\tilde{N}-1}\})\Big]\bigg\}\nonumber\\
	&=&\min\bigg\{-\eta b,\mathbb{E}\Big[\tilde{U}_{\tilde{N}-1}+\min\Big\{-\eta\max\{b,\tilde{R}_{\tilde{N}-1}\},
	-\eta\beta_1(\max\{b,\tilde{R}_{\tilde{N}-1}\})\Big\}\Big]\bigg\}\nonumber\\
  	&=&\min\bigg\{-\eta b,\frac{T}{\tilde{N}}-\eta\mathbb{E}\Big[\max\Big\{b,\tilde{R}_{\tilde{N}-1},
	\beta_1(\max\{b,\tilde{R}_{\tilde{N}-1}\})\Big\}\Big]\bigg\}\nonumber\\
	&=&\min\Big\{-\eta b,-\eta\beta_2(b)\Big\}\mbox{,}
\end{eqnarray}
where
\begin{eqnarray}
	\label{beta_2_equn}
	\beta_2(b)&=&\mathbb{E}\bigg[\max\Big\{b,R,\beta_1(\max\{b,R\})\Big\}\bigg]-\frac{T}{\eta \tilde{N}}\mbox{.}
\end{eqnarray}
$\beta_2(b)\ge\beta_1(b)$ easily follows because $\mathbb{E}\bigg[\max\Big\{b,R,\beta_1(\max\{b,R\})\Big\}\bigg]\ge\mathbb{E}\Big[\max\{b,R\}\Big]$. Next if $b\ge\alpha$ then from
Lemma \ref{beta_1_lemma}.2 and \ref{beta_1_lemma}.4  we have $\max\{b,R\}\ge\beta_1(\max\{b,R\})$ so that $\max\Big\{b,R,\beta_1(\max\{b,R\})\Big\}=\max\{b,R\}$. Therefore,
\begin{eqnarray}
 \label{beta_equal_equn}
 \beta_2(b)&=&\mathbb{E}\Big[\max\{b,R\}\Big]-\frac{T}{\eta\tilde{N}}\nonumber\\
&=&\beta_1(b)\mbox{.}
\end{eqnarray}
Hence we have shown that the Lemma  holds for $\tilde{N}-2$. 
Suppose that the Lemma (i.e., (\ref{Js_k_equn}), (\ref{beta_k_equn}) 
 and the ordering property) holds for some $k$, $1<k\le\tilde{N}-1$, then following 
the same arguments which were used to obtain (\ref{Js_last_2_equn}) and 
(\ref{beta_2_equn}) (replace $\tilde{N}-2$ by $k-1$ and $\tilde{N}-1$ by $k$) 
we can show that (\ref{Js_k_equn}) and (\ref{beta_k_equn}) hold for stage $k-1$ as well.
The ordering property can be easily shown to hold for stage $k-1$ by using the ordering property
for stage $k$.
\hfill $\blacksquare$


\section{One-Hop Performance for Different Reward Distributions ($f_R$) and Initial Beliefs ($p_0$)}
\label{simulation_cases_appendix}
In Section~\ref{one_hop_performance} we performed simulations to 
compare the one-hop performance of the various policies (recall the description 
of the implemented policies from Section~\ref{implemented_policies}).
There we had considered the context of geographical forwarding (which was the primary
motivation for our work), so that the reward associated with a relay 
is the progress it makes towards the sink (see Fig.~\ref{forwarding_set_figure}). Also the initial belief we had
considered was truncated Poisson (of mean $\lambda=10$) with $K=50$ 
(recall that $K$ is the bound on the number of relays).
From Tables~\ref{table:result_Z} and 
\ref{table:result_D} we were able to draw
the following conclusions: For large values of target $\gamma$,
\begin{itemize}
\item  The average delay of $\pi_{A-SIMPL}$ and $\pi_{INNER}$
is close to $\pi_{COMDP}$, which is the optimal policy.
\item The difference in the delays, incurred by $\pi_{INNER}$ and $\pi_{OUTER}$, is small.
\item $\pi_{A-COMPD}$ incurs a larger delay.
\end{itemize}
For smaller values of target $\gamma$, we see that all the
policies incur similar average delay.

In this appendix, to comment more on these conclusions, we have performed simulations
for few other examples, with different pairs of reward distributions ($f_R$) 
and intial beliefs ($p_0$). In each of these examples, the good performance of the
policy $\pi_{A-SIMPL}$ is observed.
We have fixed $T=1$ and normalize the rewards to take values within the interval $[0,1]$ 
for all the examples. 
The first two examples extend the scenario of geographical forwarding mentioned earlier while
in the next two we simply take $R$ to have uniform and truncated Gaussian distributions, respectively.  
As in Section~\ref{one_hop_performance} we discretize the 
state space and approximate all the pmfs with pdfs in simulations.
For each example we tabulate the results (i.e., average reward values for few values
of target $\gamma$ and the corresponding average delays) which have the same
explanation as for the Tables~\ref{table:result_Z} and \ref{table:result_D} 
(see the explanation following Tables~\ref{table:result_Z} and \ref{table:result_D}).


\subsection*{EXAMPLE 1}
\begin{itemize}
\item Reward: We consider the same scenario of geographical forwarding as in Section~\ref{one_hop_performance}.
 Here we allow the reward to be a function of the progress. 
 Let $Z_i$ be the progress made by relay $i$. Small values of $Z_i$ are not favourable because
 the packet does not make significant progress towards the sink. On the other hand when $Z_i$ is large, the 
 attenuation of the signal transmitted from the source to the relay will be large. This means a higher power is 
 required to achieve a given packet error rate. Thus, we want to penalize both small and large values of $Z_i$.
 This motivates us to choose the reward function to be $R_i=-a_1 Z_i\log(\frac{Z_i}{a_2})$. 
 $R_i$ is maximum at $Z_i=\frac{a_2}{e}$. We have choosen $a_2=0.4e$. $a_1$ is a constant used to
 normalize the maximum reward value to $1$.  Using $f_Z$ in (\ref{equn:progress_pmf}) one can obtain $f_R$.

\item Initial Belief: Bound on the number of relays is $K=40$. Initial belief is truncated Poisson with parameter $5$
{i.e.,} for $n=1,2,\cdots,K$, $p_0(n)=c\frac{5^n}{n!}e^{-5}$, where $c$ is the normalization constant. 
\end{itemize}
Results are tabulated in Tables~\ref{table:R_result_case1} and \ref{table:D_result_case1}.

\begin{table}[h!]
\centering
\begin{tabular}{|c |c |c |c |c|c|}
\hline
Target $\gamma$ & 0.7800 & 0.8200 & 0.8600 & 0.9000 & 0.9400\\
\hline
$\mathbb{E}[R_{\pi_{COMDP}}]$ & 0.7751 & 0.8164 & 0.8628 & 0.9018 & 0.9402\\
\hline
$\mathbb{E}[R_{\pi_{INNER}}]$ & 0.7755 & 0.8195 & 0.8663 & 0.8991 & 0.9397\\
\hline
$\mathbb{E}[R_{\pi_{OUTER}}]$ & 0.7826 & 0.8216 & 0.8593 & 0.9006 & 0.9407\\
\hline
$\mathbb{E}[R_{\pi_{A-COMDP}}]$ & 0.7730 & 0.8184 & 0.8651 & 0.8986 & 0.9401\\
\hline
$\mathbb{E}[R_{\pi_{A-SIMPL}}]$ & 0.7755 & 0.8166 & 0.8589 & 0.8992 & 0.9406\\
\hline
\end{tabular}
\caption{\emph{EXAMPLE 1}: Target $\gamma$ and corresponding average rewards}
\label{table:R_result_case1}
\vspace*{-8mm}
\end{table}

\begin{table}[h!]
\centering
\begin{tabular}{|c |c |c |c |c|c|}
\hline
Target $\gamma$ & 0.7800 & 0.8200 & 0.8600 & 0.9000 & 0.9400\\
\hline
$\mathbb{E}[D_{\pi_{COMDP}}]$ & 0.1950 & 0.2082 & 0.2340 & 0.2715 & 0.3649\\
\hline
$\mathbb{E}[D_{\pi_{INNER}}]$ & 0.1963 & 0.2150 & 0.2471 & 0.2839 & 0.4005\\
\hline
$\mathbb{E}[D_{\pi_{OUTER}}]$ & 0.1993 & 0.2168 & 0.2431 & 0.2865 & 0.4078\\
\hline
$\mathbb{E}[D_{\pi_{A-COMDP}}]$ & 0.1963 & 0.2164 & 0.2499 & 0.2871 & 0.4153\\
\hline
$\mathbb{E}[D_{\pi_{A-SIMPL}}]$ & 0.1963 & 0.2133 & 0.2411 & 0.2840 & 0.4056\\
\hline
\end{tabular}
\caption{\emph{EXAMPLE 1}: Average delays corresponding to average rewards in Table~\ref{table:R_result_case1}}
\label{table:D_result_case1}
\end{table}


\subsection*{EXAMPLE 2}
\begin{itemize}
 \item Reward: Again we consider the scenario of geographical forwarding. 
Let $Z_i$ be the progress made by a relay $i$ and $H_i$ be the (normalized) data rate 
from the source to the relay $i$. $H_i$ is a random variable 
which takes values from the set $\{0.2, 0.4, 0.6, 0.8, 1.0\}$. For small (large)
values of $Z_i$ there is a high (low) probability that the data rates are good.
Thus as $Z_i$ increases we want the probability of $H_i$ taking larger values to decrease.
Therefore when $Z_i=z$ we set $\mathbb{P}(H_i=h|Z_i=z)=a_zhe^{-dzh}$ 
for $h\in\{0.2,\cdots,1.0\}$. $a_zhe^{-dzh}$, as a function of $h$, attains maximum at
$\frac{1}{dz}$ so that as $Z_i$ increases $H_i$ takes lower values with high probability.   
We have choosen $d=10$.
$a_z$ is a constant to normalize the
total probability to $1$. Finally the reward associated with relay $i$ is 
$R_i=c_1 Z_i+c_2 H_i$. We choose $c_1=c_2=0.5$. 

\item Initial Belief: $K=30$ and $p_0$ is binomial with parameter $0.5$ i.e.,
 for $n=1,2,\cdots,K$, $p_0(n)= c{K\choose n}0.5^n$ where $c$ is the normalization contant. 
Such an initial belief is appropriate if initially during deployment the source had $K$ potential relays 
and at the time when the source has a packet (which happens after a significant amount of time because the events are rare),
the probability with which a relay has not failed is $0.5$ (we have ignored the case where all the relays have failed). 
\end{itemize}
Results are tabulated in Tables~\ref{table:R_result_case2} and \ref{table:D_result_case2}.

\begin{table}[h!]
\centering
\begin{tabular}{|c |c |c |c |c|c|}
\hline
Target $\gamma$ & 0.4500 & 0.5000 & 0.5500 & 0.6000 & 0.6500\\
\hline
$\mathbb{E}[R_{\pi_{COMDP}}]$ & 0.4510 & 0.5058 & 0.5489 & 0.6002 & 0.6500\\
\hline
$\mathbb{E}[R_{\pi_{INNER}}]$ & 0.4510 & 0.5050 & 0.5508 & 0.6013 & 0.6503\\
\hline
$\mathbb{E}[R_{\pi_{OUTER}}]$ & 0.4510 & 0.5056 & 0.5506 & 0.5989 & 0.6500\\
\hline
$\mathbb{E}[R_{\pi_{A-COMDP}}]$ & 0.4510 & 0.5066 & 0.5500 & 0.6002 & 0.6500\\
\hline
$\mathbb{E}[R_{\pi_{A-SIMPL}}]$ & 0.4510 & 0.5088 & 0.5518 & 0.6013 & 0.6499\\
\hline
\end{tabular}
\caption{\emph{EXAMPLE 2}: Target $\gamma$ and corresponding average rewards}
\label{table:R_result_case2}
\vspace*{-8mm}
\end{table}

\begin{table}[h!]
\centering
\begin{tabular}{|c |c |c |c |c|c|}
\hline
Target $\gamma$ &  0.4500 & 0.5000 & 0.5500 & 0.6000 & 0.6500\\
\hline
$\mathbb{E}[D_{\pi_{COMDP}}]$ & 0.0638 & 0.0830 & 0.1179 & 0.2075 & 0.4246\\
\hline
$\mathbb{E}[D_{\pi_{INNER}}]$ & 0.0638 & 0.0835 & 0.1218 & 0.2155 & 0.4456\\
\hline
$\mathbb{E}[D_{\pi_{OUTER}}]$ & 0.0638 & 0.0836 & 0.1225 & 0.2159 & 0.4582\\
\hline
$\mathbb{E}[D_{\pi_{A-COMDP}}]$ & 0.0638 & 0.0843 & 0.1213 & 0.2146 & 0.4510\\
\hline
$\mathbb{E}[D_{\pi_{A-SIMPL}}]$ & 0.0638 & 0.0858 & 0.1225 & 0.2159 & 0.4443\\
\hline
\end{tabular}
\caption{\emph{EXAMPLE 2}: Average delays corresponding to average rewards in Table~\ref{table:R_result_case2}}
\label{table:D_result_case2}
\end{table}


\subsection*{EXAMPLE 3}
\begin{itemize}
 \item Reward: $R$ is distributed uniformly on $[0,1]$.
 \item Initial Belief: $K=20$ and $p_0$ is binomial with parameter $0.5$.
\end{itemize}
Results are tabulated in Tables~\ref{table:R_result_case3} and \ref{table:D_result_case3}.

\begin{table}[h!]
\centering
\begin{tabular}{|c |c |c |c |c|c|}
\hline
Target $\gamma$ & 0.7000 & 0.7500 & 0.8000 & 0.8500 & 0.9000\\
\hline
$\mathbb{E}[R_{\pi_{COMDP}}]$ & 0.7093 & 0.7566 & 0.8030 & 0.8503 & 0.9000\\
\hline
$\mathbb{E}[R_{\pi_{INNER}}]$ & 0.7102 & 0.7588 & 0.7984 & 0.8512 & 0.9001\\
\hline
$\mathbb{E}[R_{\pi_{OUTER}}]$ & 0.7099 & 0.7523 & 0.8004 & 0.8500 & 0.8999\\
\hline
$\mathbb{E}[R_{\pi_{A-COMDP}}]$ & 0.7135 & 0.7580 & 0.8040 & 0.8488 & 0.9001\\
\hline
$\mathbb{E}[R_{\pi_{A-SIMPL}}]$ & 0.7119 & 0.7538 & 0.7968 & 0.8485 & 0.9009\\
\hline
\end{tabular}
\caption{\emph{EXAMPLE 3}: Target $\gamma$ and corresponding average rewards}
\label{table:R_result_case3}
\vspace*{-8mm}
\end{table}

\begin{table}[h!]
\centering
\begin{tabular}{|c |c |c |c |c|c|}
\hline
Target $\gamma$ & 0.7000 & 0.7500 & 0.8000 & 0.8500 & 0.9000\\
\hline
$\mathbb{E}[D_{\pi_{COMDP}}]$ & 0.1557 & 0.1846 & 0.2279 & 0.3033 & 0.5115\\
\hline
$\mathbb{E}[D_{\pi_{INNER}}]$ & 0.1588 & 0.1910 & 0.2288 & 0.3180 & 0.5443\\
\hline
$\mathbb{E}[D_{\pi_{OUTER}}]$ & 0.1594 & 0.1870 & 0.2339 & 0.3201 & 0.5515\\
\hline
$\mathbb{E}[D_{\pi_{A-COMDP}}]$ & 0.1610 & 0.1909 & 0.2367 & 0.3136 & 0.5995\\
\hline
$\mathbb{E}[D_{\pi_{A-SIMPL}}]$ & 0.1600 & 0.1872 & 0.2279 & 0.3107 & 0.5529\\
\hline
\end{tabular}
\caption{\emph{EXAMPLE 3}: Average delays corresponding to average rewards in Table~\ref{table:R_result_case3}}
\label{table:D_result_case3}
\vspace*{-8mm}
\end{table}

\subsection*{EXAMPLE 4}
\begin{itemize}
 \item Reward: Truncated Gaussian of mean $0.5$ and variance $1$ i.e., 
for $r\in[0,1]$, $f_R(r)=\frac{c}{\sqrt{2\pi}}e^{\frac{(r-0.5)^2}{2}}$ where
 $c$ is the normalization constant. 
 \item Initial Belief: $K=15$ and $p_0$ is uniform on $\{1,2,\cdots,K\}$.
\end{itemize}
Results are tabulated in Tables~\ref{table:R_result_case4} and \ref{table:D_result_case4}.

\begin{table}[h!]
\centering
\begin{tabular}{|c |c |c |c |c|c|}
\hline
Target $\gamma$ & 0.6400 & 0.6800 & 0.7200 & 0.7600 & 0.8000\\
\hline
$\mathbb{E}[R_{\pi_{COMDP}}]$ & 0.6500 & 0.6725 & 0.7208 & 0.7625 & 0.8001\\
\hline
$\mathbb{E}[R_{\pi_{INNER}}]$ & 0.6487 & 0.6728 & 0.7240 & 0.7601 & 0.7997\\
\hline
$\mathbb{E}[R_{\pi_{OUTER}}]$ & 0.6388 & 0.6807 & 0.7213 & 0.7600 & 0.7998\\
\hline
$\mathbb{E}[R_{\pi_{A-COMDP}}]$ & 0.6259 & 0.6791 & 0.7225 & 0.7618 & 0.7997\\
\hline
$\mathbb{E}[R_{\pi_{A-SIMPL}}]$ & 0.6302 & 0.6769 & 0.7146 & 0.7607 & 0.8009\\
\hline
\end{tabular}
\caption{\emph{EXAMPLE 4}: Target $\gamma$ and corresponding average rewards}
\label{table:R_result_case4}
\vspace*{-8mm}
\end{table}

\begin{table}[h!]
\centering
\begin{tabular}{|c |c |c |c |c|c|}
\hline
Target $\gamma$ & 0.6400 & 0.6800 & 0.7200 & 0.7600 & 0.8000\\
\hline
$\mathbb{E}[D_{\pi_{COMDP}}]$ & 0.2092 & 0.2222 & 0.2600 & 0.3060 & 0.3799\\
\hline
$\mathbb{E}[D_{\pi_{INNER}}]$ & 0.2274 & 0.2473 & 0.2981 & 0.3478 & 0.4386\\
\hline
$\mathbb{E}[D_{\pi_{OUTER}}]$ & 0.2307 & 0.2622 & 0.3031 & 0.3576 & 0.4460\\
\hline
$\mathbb{E}[D_{\pi_{A-COMDP}}]$ & 0.2290 & 0.2689 & 0.3122 & 0.3740 & 0.4735\\
\hline
$\mathbb{E}[D_{\pi_{A-SIMPL}}]$ & 0.2218 & 0.2577 & 0.2924 & 0.3532 & 0.4473\\
\hline
\end{tabular}
\caption{\emph{EXAMPLE 4}: Average delays corresponding to average rewards in Table~\ref{table:R_result_case4}}
\label{table:D_result_case4}
\end{table}

\end{document}